\documentclass[pre,twocolumn,showpacs,superscriptaddress,preprintnumbers,floatfix]{revtex4-1}

\usepackage{etex}
\usepackage{ifpdf}
\usepackage[pagebackref]{hyperref}
\usepackage{dcolumn}
\usepackage{url}
\usepackage{mathtools}
\usepackage{amsmath}
\usepackage{amsfonts}
\usepackage{amssymb}
\usepackage{amsthm}
\usepackage{array}
\usepackage{bm}   % bold math
\usepackage{bbm}
\usepackage{verbatim}
\usepackage{stmaryrd}
\usepackage{xcolor}
\usepackage{setspace}
\usepackage{graphicx}

\newif \ifgenerate
\newif \ifcache
\cachetrue

\generatefalse

\def\baseimagedir{}
\def\tikzpath{\baseimagedir tikz/}
\def\tikzpathexternal{}

\ifgenerate
  \usepackage{tikz}
  \input{\tikzpath tikzlibraries}
  \input{\tikzpath vaucanson.tikz}
  \ifcache
    \usetikzlibrary{external}
  \fi
\fi

\providecommand{\Symbol}[1]{\textcolor{blue}{#1}}

\definecolor{FCSA}{RGB}{141,211,199}
\definecolor{FCSB}{RGB}{255,255,179}
\definecolor{RCSC}{RGB}{185,138,196}
\definecolor{RCSD}{RGB}{231,143,111}
\definecolor{RCSE}{RGB}{128,177,211}

\theoremstyle{plain}	\newtheorem{Lem}{Lemma}
\theoremstyle{plain} 	\newtheorem{Cor}{Corollary}
\theoremstyle{plain} 	\newtheorem{The}{Theorem}
\theoremstyle{plain} 	
\theoremstyle{plain} 	
\theoremstyle{plain}	
\theoremstyle{plain}	\newtheorem{Def}{Definition}
\theoremstyle{plain}	

\makeatletter
\renewenvironment{proof}[1][\proofname]{%
\par\pushQED{\qed}\normalfont%
\topsep6\p@\@plus6\p@\relax\trivlist%
\item[\hskip\labelsep\bfseries#1\@addpunct{.}]\itshape\ignorespaces}{%
\popQED\endtrivlist\@endpefalse}%
\makeatother

\def\clap#1{\hbox to 0pt{\hss#1\hss}}

\def\mathclap{\mathpalette\mathclapinternal}

\def\mathclapinternal#1#2{%
\clap{$\mathsurround=0pt#1{#2}$}}

\newcommand{\emphasis}[1] {\emph{#1}}
\newcommand{\vocab}[1] {\emph{#1}}

\newcommand{\eM}     {\mbox{$\epsilon$-machine}}
\newcommand{\eMs}    {\mbox{$\epsilon$-machines}}
\newcommand{\EM}     {\mbox{$\epsilon$-Machine}}
\newcommand{\EMs}    {\mbox{$\epsilon$-Machines}}

\newcommand{\Process}{\mathcal{P}}

\newcommand{\MeasAlphabet}	{\mathcal{A}}
\newcommand{\MeasSymbol}   { {X} }
\newcommand{\meassymbol}   { {x} }

\newcommand{\past}	{ \smash{\overleftarrow {\meassymbol}} }

\newcommand{\future}	{ \smash{\overrightarrow{\meassymbol}} }

\newcommand{\CausalState}	{ \mathcal{S} }

\newcommand{\causalstate}	{ \sigma }
\newcommand{\CausalStateSet}	{ \boldsymbol{\CausalState} }
\newcommand{\AlternateState}	{ R }
\newcommand{\AlternateStateSet}	{ \mathcal{R} }

\newcommand{\Prob}      {\Pr} % use standard command

\newcommand{\Cmu}		{C_\mu}
\newcommand{\hmu}		{h_\mu}
\newcommand{\EE}		{{\bf E}}

\newcommand{\PC}		{\chi}

\newcommand{\forward}{+}
\newcommand{\reverse}{-}
\newcommand{\forwardreverse}{\pm} % \pm

\newcommand{\ForwardDFA}{D^+}
\newcommand{\ReverseDFA}{D^-}
\newcommand{\FutureCausalState}	{ {\CausalState}^{\forward} }
\newcommand{\futurecausalstate}	{ \sigma^{\forward} }

\newcommand{\PastCausalState}	{ {\CausalState}^{\reverse} }

\newcommand{\ForwardCausalState}	{ {\CausalState}^{\forward} }
\newcommand{\forwardcausalstate}	{ \sigma^{\forward} }

\newcommand{\ReverseCausalState}	{ {\CausalState}^{\reverse} }
\newcommand{\reversecausalstate}	{ \sigma^{\reverse} }
\newcommand{\BiCausalState}		{ {\CausalState}^{\forwardreverse} }

\newcommand{\ForwardCausalStateSet}	{ {\CausalStateSet}^{\forward} }
\newcommand{\ReverseCausalStateSet}	{ {\CausalStateSet}^{\reverse} }

\newcommand{\BiCausalStateSet}	{ {\CausalStateSet}^{\forwardreverse} }
\newcommand{\eMachine}	{ M }
\newcommand{\FutureEM}	{ {\eMachine}^{\forward} }

\newcommand{\ForwardEM}	{ {\eMachine}^{\forward} }
\newcommand{\ReverseEM}	{ {\eMachine}^{\reverse} }

\newcommand{\FutureCmu}	{ C_\mu^{\forward} }
\newcommand{\PastCmu}	{ C_\mu^{\reverse} }
\newcommand{\ForwardCmu}	{ C_\mu^{\forward} }
\newcommand{\ReverseCmu}	{ C_\mu^{\reverse} }

\newcommand{\one}{\mathbf{1}}

\newcommand{\lastindex}[2]{
  \edef\tempa{0}
  \edef\tempb{#2}
  \ifx\tempa\tempb
    \edef\tempc{#1}
  \else
    \edef\tempa{0}
    \edef\tempb{#1}
    \ifx\tempa\tempb
      \edef\tempc{#2}
    \else
      \edef\tempc{#1+#2}
    \fi
  \fi
  \tempc
}

\newcommand{\GI}{\varphi}
\newcommand{\OI}{\zeta}
\newcommand{\ORDER}{\text{order\nobreakdash-}}

\newcommand{\MOrder}{\ensuremath{R}}

\newcommand{\CSjoint}[1][,]{
   \edef\tempa{:}
   \edef\tempb{#1}
   \ifx\tempa\tempb
      \ensuremath{\FutureCausalState\!#1\PastCausalState}
   \else
      \ensuremath{\FutureCausalState#1\PastCausalState}
   \fi
}

\newif\ifpm 
\edef\tempa{\forwardreverse}
\edef\tempb{\pm}
\ifx\tempa\tempb
   \pmfalse
\else
   \pmtrue  
\fi

\renewcommand{\Pr}{\mathbb{P}}
\newcommand{\Measure} {\mu}

\begin{document}

\title{Information Symmetries in Irreversible Processes}

\author{Christopher J. Ellison}
\email{cellison@cse.ucdavis.edu}
\affiliation{Complexity Sciences Center\\
Physics Department, University of California at Davis,\\
One Shields Avenue, Davis, CA 95616}

\author{John R. Mahoney}
\email{jmahoney3@ucmerced.edu}
\affiliation{School of Natural Sciences,
University of California at Merced,\\
5200 North Lake Road, Merced, CA 95343}

\author{Ryan G. James}
\email{rgjames@ucdavis.edu}
\affiliation{Complexity Sciences Center\\
Physics Department, University of California at Davis,\\
One Shields Avenue, Davis, CA 95616}

\author{James P. Crutchfield}
\email{chaos@cse.ucdavis.edu}
\affiliation{Complexity Sciences Center\\
Physics Department, University of California at Davis,\\
One Shields Avenue, Davis, CA 95616}
\affiliation{Santa Fe Institute,
1399 Hyde Park Road, Santa Fe, NM 87501}

\author{J\"{o}rg Reichardt}
\email{jreichardt@ucdavis.edu}
\affiliation{Complexity Sciences Center\\
Physics Department, University of California at Davis,\\
One Shields Avenue, Davis, CA 95616}

\date{\today}
\bibliographystyle{unsrt}

\def\ourAbstract{%
We study dynamical reversibility in stationary stochastic processes
from an information theoretic perspective. Extending earlier work on the
reversibility of Markov chains, we focus on finitary processes with arbitrarily
long conditional correlations. In particular, we examine stationary processes
represented or generated by edge-emitting, finite-state hidden Markov models.
Surprisingly, we find pervasive temporal asymmetries in the statistics of such
stationary processes with the consequence that the computational resources
necessary to generate a process in the forward and reverse temporal directions
are generally not the same. In fact, an exhaustive survey indicates that most
stationary processes are irreversible.  We study the ensuing relations between
model topology in different representations, the process's statistical
properties, and its reversibility in detail. A process's temporal asymmetry is
efficiently captured using two canonical unifilar representations of the
generating model, the forward-time and reverse-time
\texorpdfstring{\eMs}{epsilon-machines}. We analyze example irreversible
processes whose \texorpdfstring{\eM}{epsilon-machine} presentations change
size under time reversal, including one which has a finite number of
recurrent causal states in one direction, but an infinite number in the
opposite. From the forward-time and reverse-time
\texorpdfstring{\eMs}{epsilon-machines}, we are able to construct a symmetrized,
but nonunifilar, generator of a process---the bidirectional machine.  Using
the bidirectional machine, we show how to directly calculate a process's
fundamental information properties, many of which are otherwise only poorly
approximated via process samples.
The tools we introduce and the insights we offer provide a better understanding
of the many facets of reversibility and irreversibility in stochastic processes.
}

\def\ourKeywords{%
  stochastic process, reversibility, irreversibility, hidden Markov model,
  Markov chain, information diagram, presentation, bidirectional machine,
  \texorpdfstring{\eM}{epsilon-machine}
}

\hypersetup{
  pdfauthor={C. J. Ellison, J. R. Mahoney, R. G. James,
             J. P. Crutchfield, and J. Reichardt},
  pdftitle={Information Symmetries in Irreversible Processes},
  pdfsubject={\ourAbstract},
  pdfkeywords={\ourKeywords},
  pdfproducer={},
  pdfcreator={}
}

% ************************* ABSTRACT *************************
\begin{abstract}
\ourAbstract

\vspace{0.1in}

\end{abstract}

\preprint{Santa Fe Institute Working Paper 11-06-XXX}
\preprint{arxiv.org:1106.XXXX [physics.gen-ph]}

\maketitle
\
% ****************************************************************

%\tableofcontents

\setstretch{1.1}

\vspace{0.3in}

{\bf One of the principal early mysteries of thermodynamics was the origin of
  irreversibility: While microscopic equations of motion describe behaviors that
  are the same in both time directions, why do large-scale systems exhibit
  temporal asymmetries? Many thermodynamic processes go in one direction: Closed
  systems devolve from order to disorder, heat flows from high temperature to
  low temperature, and shattered glass does not reassemble itself spontaneously.
  These are described as transient relaxation processes in which a system moves
  from one macroscopic state to another with high probability since there is
  an overwhelming number of microscopic configurations that realize the
  eventual state.
  
  Here, we analyze a generalized notion of irreversibility:
  Behavior in reverse time gives rise to a different stochastic process than
  that in forward time. This \emph{dynamical irreversibility} subsumes
  relaxation, but is not so constrained, since it can occur in a
  nonequilibrium steady state. A dynamical parallel to the shattered glass
  example of transient relaxation is found in a ``continuous-flow" glass
  grinder: Continuously fed whole glass, the grinder eventually produces
  glass pieces that are sufficiently small to pass out via a sieve. After
  a transient start-up time, the distribution of glass sizes settles down to
  a steady state. The glass grinding process is dynamically irreversible.

  We explore irreversibility in stationary stochastic systems using new tools from
  information theory and computational mechanics. We show that a system's causal
  structure and information storage depend on time's arrow, while its rate of
  generating information does not. We develop a time-symmetric
  representation---the bi-directional machine---that allows one to directly
  determine key informational and computational properties, including how much
  stored information is hidden from observation, the number of excess
  statistical degrees of freedom, the amount of internal information that
  anticipates future behavior, and the like. We summarize the analysis via a
  new irreversibility classification scheme for stochastic processes. Overall,
  the result is an enriched view of irreversibility and its companion
  properties---a view that enhances our understanding of the relationship
  between energy and information and of the structure of the physical
  substrates that carry them.
} % End CHAOS Lead Paragraph bold font

\section{Introduction}

Dynamical systems, by definition, evolve in time. Practically all of what we
may know about a system is derived from careful observation of its change in
time. In their attempt to understand underlying mechanisms, physicists cast
observations in the language of mathematics, spelling out
``equations of motion'' to model how a system's temporal behavior
arises from the forces acting on it. In some settings, such modeling
allows for forecasting a system's behavior given its current state,
but also allows for tracing its evolution backward in time.

The equations of motion of classical mechanics meet this ideal;
they are \vocab{dynamically reversible}. From current observations of
the sky, we are able to precisely determine planet motions
hundreds of years into the past and future; we can
determine the future course of meteorites, but also where they
came from. This dynamical reversibility is tied to the fact that
the mechanical equations of motion provide an invertible one-to-one
mapping of a system's current state to its future state; that is,
they specify a deterministic dynamic. Given a mechanical system's
current state at a single instant, Laplace's Daemon, in principle, can
predict the system's entire future and entire past \cite{Lapl52}.

In practice, the limited precision to which we can
specify initial conditions and the often high sensitivity of the
equations of motion with respect to changes of the initial
conditions restrict our ability to predict a system's future or
to reconstruct its history over a long period of time.

Acknowledging this fundamental limitation, statistical mechanics introduced
the distinction between a system's macroscopic state and its microscopic state.
For example, the precise momenta and positions of the particles in a gas
container form the system's microstate, whose behavior is governed by
deterministic, reversible dynamics. Only averages over these microstates are
accessible to us, though, being measured as pressure, temperature, volume, and
the like. These thermodynamic-state variables in turn describe the system's
macrostate and their interdependence is given by the thermodynamic equations of
state. Interestingly, once a thermodynamic system reaches equilibrium---and
only then are the thermodynamic-state variables defined and the state
equations valid---we have no way to discover the system's past.
That is, we cannot know how the system reached this thermodynamic state since
it is not possible to trace back the system's
evolution from its current state: Its macroscopic dynamics are irreversible
\footnote{Note that thermodynamics does speak of reversible macroscopic
processes. Consider, for example, a common thermodynamic process: the isobaric
expansion of a gas from macrostate $A$ with volume $V_A$,
temperature $T_A$, and pressure $p$ to macrostate $B$ with volume $V_B>V_A$
and temperature $T_B>T_B$ at pressure $p$. This process is called reversible,
if there exists a way of manipulating the gas back to macrostate $A$,
once it is found in macrostate $B$, such as by cooling it. This notion of
thermodynamic reversibility, however, differs from our notion of
reversibility which focuses on the ability to reconstruct a history from
observations.}. Equilibrium thermodynamics then leaves us with a
description of a system in terms of thermodynamic macrostates
that are entirely devoid of traces of the system's past and that are
trivial with respect to the system's further evolution.

Between the extremes of deterministic mechanical systems with their complete
reversibility and thermodynamic systems that do not admit reconstructing the
system's past, we find stochastic processes. Stochastic processes exhibit
nontrivial, nondeterministic dynamical evolution that combines the ability to
reconstruct historical evolution and to forecast future behavior in a
probabilistic setting.

Here, developing an information-theoretic perspective, we study the
reversibility of stochastic processes; specifically, our ability to make
assertions about a process's past from current and future observations. We
contrast the act of reconstructing a process's past based on current and future
observations (retrodiction) with that of forecasting a process's future based
on past and current observations (prediction). We show the two tasks exhibit
a number of unexpected and nontrivial asymmetries. In particular, we show
that in contrast to deterministic dynamical systems, where the forward and
reverse evolution can be computed at the same computational cost---solving a
differential equation---predicting and retrodicting a stochastic process's
evolution may come at very different computational costs. More precisely, we
show that the canonical generators of stochastic processes, their ``equations
of motion'' so to speak, are generally far from invariant under time reversal.
Via an exhaustive survey, we demonstrate that irreversibility is an
overwhelmingly dominant property of structurally complex stochastic processes.
This asymmetry shows that depicting processes only by either their
forward or reverse generators typically does not provide a complete
description. This leads us to introduce a time-symmetric representation of a
stochastic process that allows a direction calculation of key
informational and computational quantities associated with the process's
evolution in forward and backward time directions. With these tools at hand, we
are then able to establish a novel classification of stochastic
processes in terms of their reversibility, providing new
insights into the diversity of information processing embedded in
physical systems.

Of fundamental importance for our discussion is the notion of the ``state''
of a probabilistic process and the use of state-based models---the so-called
generators---to describe stochastic processes. We introduce these in
Section~\ref{sec:ProcStateGen} . Continuing in more familiar territory,
Section~\ref{sec:ObservableStates} reviews reversibility in processes whose
generators have states which can be directly observed---the so-called Markov
chains. Section~\ref{sec:hmm} expands the discussion to a broader class of
models, the hidden Markov models (HMMs), whose states cannot be directly
observed. There, we utilize the information measures from Ref.~\cite{Crut10a} to
describe ways in which a process hides internal structure from observations.
Then we draw out the differences between models of processes with and without
observable states. In this, we confront the issue of process structure. This
leads Sec.~\ref{sec:structure} to introduce a canonical representation for
each process---the \eM. At this point, irreversibility of HMMs becomes
necessarily tied to properties of the \eM. There, we introduce the \eM\
information diagram which is a useful roadmap for the various information
measures and corresponding process properties. A number of example processes
are analyzed to help ground the concepts introduced up to this point. A new
presentation is required to go further, however, and Sec. \ref{sec:bieM}
introduces and analyzes a process's bidirectional machine using Ref. \cite{Crut10a}'s information measures. Finally, we conclude by drawing out the
thermodynamic implications for these notions of irreversibility and commenting
on its role in applications.

\section{Processes and Generators}
\label{sec:ProcStateGen}

To keep our analysis of irreversibility constructive, our focus
here is on discrete-time, discrete-valued stationary processes and
their various alternate representations. This class includes the symbolic
dynamics of chaotic dynamical systems, one-dimensional spin chains, and
cellular automata spatial configurations, to mention three well-known,
complex applications.
Historically, one-dimensional stochastic processes were studied using
\vocab{generators}---models that reproduce the process's statistics in a
time-ordered sequence. The tradition of using generators is so strong that
their time-order is often treated as synonymous with the process's time-order
which, as the following will remind the reader, need not exist. Much of the following
requires that we loosen the seemingly natural assumption of time-order.

To begin, we define processes strictly in terms of probability
spaces~\cite{Uppe97a}. Consider the space $\MeasAlphabet^\mathbb{Z}$ of
bi-infinite sequences consisting of symbols from $\MeasAlphabet$, a finite
set known as the \vocab{alphabet}.
Taking $\mathbb{X}$ to be the $\sigma$-field generated by the
\vocab{cylinder sets} of $\MeasAlphabet^\mathbb{Z}$, we assign probabilities
to sets in $\mathbb{X}$ via a measure $\Measure$. The $3$-tuple
$(\MeasAlphabet^\mathbb{Z}, \mathbb{X}, \Measure)$ is a
\vocab{probability space} that we refer to as a \vocab{process},
denoting it $\Process$.

Let $\MeasSymbol_i$ denote the random variable that describes the outcomes
at index $i$. As a convenient shorthand~\footnote{This is equivalent to index
notation in the Python programming language.}, we denote random variable blocks
as $\MeasSymbol_{i:j} = \MeasSymbol_i \MeasSymbol_{i+1} \cdots
\MeasSymbol_{j-1}, j \geq i$. When $j = i$, the block has length zero
and this is used to keep definitions simple.

For example, consider a process with alphabet $\MeasAlphabet = \{a,b,c\}$ for
which the word $w=abc$ has the corresponding cylinder set
$\{ x \in \MeasAlphabet^\mathbb{Z} | X_0=a, X_1=b, X_2=c \}$.
The probability of $w$ is defined to be the probability of its
cylinder set in $\mathbb{X}$:
\begin{align*}
  \Pr(X_{0:3} = w) &= \Pr(X_0=a, X_1=b, X_2=c) \\
                   &= \Measure \left(\{ x \in \MeasAlphabet^\mathbb{Z} |
                               X_0=a, X_1=b, X_2=c \}\right).
\end{align*}

Notice that \emph{time} does not appear explicitly in the definition of a
process as a probability space. Indeed, the indexing of $X_i$ can refer, for
example, to locations on a spatial lattice.

While one need not interpret a process in terms of time,
temporal interpretations are often convenient.
The random variable block leading up to ``time'' $t$ is referred to as the
\vocab{past} and denoted $X_{:t} \equiv \ldots X_{t-3} X_{t-2} X_{t-1}$.
Everything from $t$ onward is referred to as the \emph{future} and denoted
$X_{t:} \equiv X_t X_{t+1} X_{t+2} \ldots$. We restrict ourselves to
\vocab{stationary} processes by demanding that $\mathbb{P}$
yield the same probabilities for blocks whose indices are shifts of one
another: $\Pr(X_{0:L}) = \Pr(X_{t:t+L})$ for all $t$ and $L$. When considering
generative models, we work with a semi-infinite sequence of
random variables, but due to stationarity, the distribution can be
uniquely extended to a probability distribution over bi-infinite
sequences~\cite{Uppe97a}.

Generators are dynamical systems and so \emph{time}, as a concept, is
fundamental. That being said, there are two natural and, generally, distinct
ways of generating a process. When the time order of the generator coincides
with the process's index, which increases (a priori) left-to-right, the model is a
\vocab{forward generator} of the process.  When its time order is the opposite
of the process's index, the model is a \vocab{reverse generator}.

Given a process, if we isolate a block of symbols, that block's probability
is the same using the forward and reverse generators. The only difference
is in how the block indices are interpreted~\footnote{This is yet another
reminder that probability \emph{alone} cannot determine
causality~\cite{Pearl2009a}.}. On occasion, it will be helpful to consider a
random variable whose index increases when scanning right-to-left, as this
corresponds to increasing time from a reverse generator's frame of reference.
Such random variables will be decorated with a tilde as in $\widetilde{X}_t$.

One might object to this detailed level of distinction on the grounds that
different indexings mean that, in fact, we have two different
processes---processes that are coupled to one another under time reversal.
We acknowledge this point, but simplicity later on leads us to choose to refer
to \emphasis{the} process and, additionally, its forward and reverse generators.

Finally, note that our terminology---\emph{past} and \emph{future}---smacks of
privileging the process's forward generator. Indeed, the reverse generator
has its own ``past'' which corresponds to the forward generator's ``future''.
Generally though, we avoid basis-shifting discussions and continue to use the
biased terminology in prose, definitions, and figures. The result is that one
must consciously transform scanning process variables from one way to the other.
Examples will exercise this and so help clarify the issue.

\section{Generators with Observable States}
\label{sec:ObservableStates}

We review basic results about Markov processes, their reversibility, and their
models---Markov chains. In Markov chains, the states of the system are
defined to be the system observables and, so, Markov chains are models
of Markov processes whose states are observable. For a more detailed treatment
see Ref.~\cite{Levin2008}.

\subsection{Definitions}

A \vocab{finite Markov} process is a sequence of random variables
$\MeasSymbol_0 \MeasSymbol_1, \ldots$ each taking values from a finite
set $\MeasAlphabet$. However, the sequence is constrained
such that the probability of any symbol depends only on the most recently
seen symbol. Thus, for $x,y \in \MeasAlphabet$, $w \in \MeasAlphabet^{L-1}$,
and $L \in \mathbb{N}$, we have:
\begin{align*}
  \Pr( X_L=y | X_{0:L} = wx ) &= \Pr(X_L=y| X_{L-1} = x).
\end{align*}
Assuming stationarity, a finite Markov process is uniquely specified by a
right-stochastic matrix:
\begin{align*}
  T(x,y) \equiv \Pr(X_1=y|X_0=x).
\end{align*}
This matrix defines the model class of \emph{Markov chains}. As an example,
consider the \vocab{Golden Mean Process}~\cite{Crut92c}, whose
Markov chain is shown in Fig.~\ref{fig:GMP}(a). It has \vocab{state space}
$\MeasAlphabet = \{0,1\}$ and its \vocab{state transitions} are labeled by $T(x,y)$.
This Markov chain is \vocab{irreducible} since from each state one can reach any
other state by following successive transitions. We work only with
irreducible Markov chains in the following.

\begin{figure}
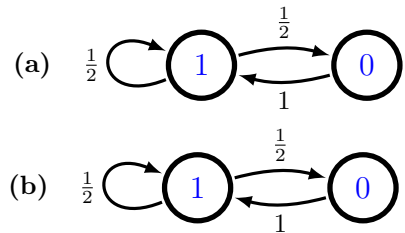

\centering

  \ifgenerate
    \ifcache
      \tikzsetnextfilename{gmp_mc}
    \fi
    \input{\tikzpathgmp_mc.tikz}
  \else
    \includegraphics{\tikzpathexternalgmp_mc}
  \fi

  \ifgenerate
    \ifcache
      \tikzsetnextfilename{gmp_mctr}
    \fi
    \input{\tikzpathgmp_mctr.tikz}
  \else
    \includegraphics{\tikzpathexternalgmp_mctr}
  \fi

\caption{(a) An irreducible Markov chain $M$ of the Golden Mean
  Process, which consists of all binary sequences with no consecutive $0$s.
  (b) Its time-reversed chain $\widetilde{M}$ which, in this case,
  is the same as the original chain.
  }
\label{fig:GMP}
\end{figure}

Every irreducible finite Markov chain has a unique
\vocab{stationary distribution} $\pi$ over $\MeasAlphabet$ obeying:
\begin{align*}
  \pi(y) = \sum_{x\in\MeasAlphabet} \pi(x) T(x,y)
  \quad \text{for all~} y \in \MeasAlphabet.
\end{align*}
In matrix notation, we simply write $\pi = \pi T$.  The Golden Mean
Markov chain has stationary distribution $\pi = (2/3, 1/3)$,
where $\Pr(X_0 = 1) = \pi(1) = 2/3$.

We calculate the probability of any word $\meassymbol_{0:L}$ by factoring
the joint probability $\Pr(\MeasSymbol_{0:L} = \meassymbol_{0:L})$ into a
product of conditional probabilities. An application of the Markov property
reduces the calculation to:
\begin{align*}
  &\hspace{-.2in}\Pr(X_0=x_0, X_1=x_1, \ldots, X_{L-1} = x_{L-1}) \\
  &= \pi(x_0) T(x_0, x_1) T(x_1, x_2)\cdots T(x_{L-2}, x_{L-1})\\
  &= \pi(x_0) \prod_{t=1}^{L-1} T(x_{t-1}, x_t).
\end{align*}

More generally, one considers \vocab{order-$R$ Markov} processes for which the
next symbol depends on the previous $R$ symbols. (See App.~\ref{app:markov}.)
Although these can be shown to be equivalent to a standard Markov chain over
a larger state space, we avoid this approach and consider the Markov order as a
property of the process. When the next symbol depends on the entire past,
though, then $R$ is infinite and the Markov chain, in effect, has an infinite
number of states. In Sec.~\ref{sec:hmm} we show how hidden Markov models can be
used to represent many such chains, while utilizing only a finite state space.

\subsection{Reversibility}

A intuitive definition of a \emph{reversible} Markov process is that it should
be indistinguishable (in probability) from the same process run backwards in
time. Thus, we define a Markov process as \vocab{reversible} if and only if for
all $w \in \MeasAlphabet^L$ and all $L \in \mathbb{N}$, we have:
\begin{align}
  \Pr(X_{0:L} = w) = \Pr(X_{0:L} = \widetilde{w}) ~,
\label{eq:revself}
\end{align}
where $w = w_0 \ldots w_{L-1}$ and $\widetilde{w} = w_{L-1} \ldots w_0$ is its
reversal.

Given a Markov process, if the transition matrix of its unique chain obeys:
\begin{align}
  \pi(x) T(x,y) = \pi(y) T(y,x),
\end{align}
for all $x,y \in \MeasAlphabet$, then we say the Markov chain is in
\vocab{detailed balance}. Note that the
uniqueness of the chain allows us to associate detailed balance with the
Markov process as well. The Markov chain representation of the Golden Mean
Process in Fig.~\ref{fig:GMP}(a) is in detailed balance.

It turns out that a stationary, finite Markov process is reversible if
and only if its Markov chain, as specified by $T$, is detailed
balance~\cite{Kelly1979a}. To see this in one direction, assume detailed
balance, then:
\begin{align*}
  \Pr(X_{0:L} = w)
  &= \pi(w_0) \prod_{t=1}^{L-1} T(w_{t-1}, w_t) \\
% &= \pi(w_0) \prod_{t=1}^{L-1} \frac{\pi(w_t) T(w_t, w_{t-1})}{\pi(w_{t-1})}\\
  &= \pi(w_{L-1}) \prod_{\mathclap{t=L-1}}^{1} T(w_t, w_{t-1})\\
  &= \Pr(X_{0:L} = \widetilde{w}) ~.
\end{align*}
Conversely, if the Markov process is reversible, then by considering only
words of length two we have $\Pr(X_0=x,X_1=y) = \Pr(X_0=y,X_1=x)$. This is
exactly the statement of detailed balance.

Given a Markov chain, we can use the condition for detailed balance to define
another chain that generates words with the same probabilities as the original
chain, but in reverse order.  If $T$ is the state transition matrix of
an irreducible Markov chain and $\pi$ is its unique stationary distribution,
then its \vocab{time-reversed} Markov chain has state transition matrix given by:
\begin{align}
  \widetilde{T}(x,y)
    %&= \Pr(\widetilde{X}_1=y | \widetilde{X}_0=x) \\
    &\equiv \Pr(X_0=y | X_1=x) \nonumber \\
    &= \frac{\pi(y) T(y,x)}{\pi(x)} ~.
\end{align}
It is easy to see that if $\pi$ is stationary for $T$, then it is also
stationary for $\widetilde{T}$. Figure~\ref{fig:GMP}(b) shows the time-reversed
chain for the Golden Mean Process. It is the same as the forward-time
chain and, thus, is also in detailed balance.

Considering the time-reversed Markov chain as a generator, we interpret:
\begin{align*}
  \pi(x)\widetilde{T}(x,y)\widetilde{T}(y,z)
\end{align*}
as the generator's probability of seeing $x$ followed by $y$ followed by $z$.
In its local time perspective, we can represent this as
$\widetilde{X}_{0:3} = xyz$.  By construction, our
expectation is that this probability should be equal to the probability (as
calculated by the forward generator) of seeing $x$ preceded by $y$ preceded
by $z$. That is, $X_{0:3} = zyx$.  And so, we can
justify the designation of being the \emph{time-reversed} Markov chain by
demonstrating that it does, indeed, generate words in reverse time:
\begin{align*}
  \Pr(\widetilde{X}_{0:L}=w)
    &= \pi(w_0) \prod_{t=1}^{L-1} \widetilde{T}(w_{t-1},w_t) \\
%   &= \pi(w_0) \prod_{t=1}^{L-1} \frac{\pi(w_t) T(w_t,w_{t-1})}{\pi(w_{t-1})}\\
    &= \pi(w_{L-1}) \prod_{\mathclap{t=L-1}}^{1} T(w_t,w_{t-1})\\
    &= \Pr(X_{0:L} = \widetilde{w}) ~.
\end{align*}
This result provides an alternative characterization of reversibility in
Markov processes: A Markov process is reversible if and only if:
\begin{align}
\label{eq:revcomp}
  \Pr(X_{0:L} = w) = \Pr(\widetilde{X}_{0:L} = w).
\end{align}
Note that while Eq.~\eqref{eq:revself} is a self-comparison test,
Eq.~\eqref{eq:revcomp} is a comparison between two distinct Markov chains. Also,
observe that if a Markov chain is reversible, then $T = \widetilde{T}$, due to
detailed balance. Thus, a reversible Markov chain is identical to its
time-reversed Markov chain. We return to this point when we define
reversibility for hidden Markov models.

What about irreversible Markov processes? A simple example will suffice.
Consider the process that generates the periodic sequence
$\ldots ABCABCABC\ldots$. Note that the time-reversed Markov chain differs:
the forward generator will emit $AB$ but not $BA$, while the
reverse generator produces $BA$ but not $AB$.

Finally, we comment briefly on the difference between a Markov process and its
associated Markov chain. The Markov process exists in the abstract, describing a
measure over bi-infinite strings. The Markov chain is a one-sided generator
representation taking the form of a single matrix. Within this class of
representations, each stationary and finite Markov process has
exactly one finite-state Markov chain. Markov processes can
\emph{also} be represented in another model class---the hidden Markov
models---and within that model class, we will see that a given Markov process
can have multiple presentations.

\section{Generators with Unobservable States}
\label{sec:hmm}

In a similar manner, we now consider models of processes whose states are not
directly observable, also known as hidden Markov models. Though rather less
well understood than Markov processes, much progress has recently been made; for
example, see Ref.~\cite{Ephr02a}. Along the way, we highlight differences
between hidden Markov models and Markov chains---differences that force one to
consider questions of structure very carefully.

\subsection{Definitions}

We begin with a Markov chain $R_0 R_1 R_2 \ldots$ over a finite state set
$\AlternateStateSet$, the \vocab{state alphabet}. This chain is internal to the
hidden Markov model. Then, a finite-state hidden Markov model (HMM) is a
sequence of outputs $X_0 X_1 X_2 \ldots$, each taking values from a finite set
$\MeasAlphabet$ that we now call the \vocab{output alphabet}. The output
sequence is generated by the internal Markov chain through a set of
\vocab{transition-output} matrices---one matrix for each symbol $x \in
\MeasAlphabet$. Each matrix element $T_x(\alpha, \beta)$ gives the transition
from (internal) state $\alpha$ to state $\beta$ on generating output $x \in
\MeasAlphabet$. That is,
\begin{align*}
  T_x(\alpha, \beta) \equiv \Pr(X_0=x, R_1=\beta | R_0=\alpha) ~.
\end{align*}
Note that the internal Markov chain's transition matrix is the marginal
distribution over the output symbol:
\begin{align*}
  T(\alpha, \beta) & = \sum_{\mathclap{x\in\MeasAlphabet}} T_x(\alpha, \beta) \\
                   & = \Pr(R_1=\beta | R_0=\alpha) ~.
\end{align*}
If, for each $x \in \MeasAlphabet$ and $\alpha \in \AlternateStateSet$ there
exists at most one $\beta \in \AlternateStateSet$ such that $T_x(\alpha, \beta) >
0$, then we say the hidden Markov model is \vocab{unifilar}. An equivalent
statement is that the entropy of the next state, conditioned on the
current state and symbol, is zero: $H[\AlternateState_1 | \AlternateState_0,
\MeasSymbol_0] = 0$.

The \emphasis{hidden} aspect of a hidden Markov model refers to the fact that the
internal Markov chain is not directly observed---only the sequence of output
symbols $X_0 X_1 X_2 \ldots$ is seen. Note, the process associated with a hidden
Markov model refers \emph{only} to the probability distribution
$\Prob(\ldots X_0 X_1 X_2 \ldots)$ over the output symbols $\MeasSymbol_t$
and \emphasis{not} over the joint process $(R_t,X_t)$.

Non-Markov processes differ from Markov processes in that they exhibit
arbitrarily long conditional correlations. That is, the probability of the next
symbol may depend on the \emphasis{entire history} leading up to this symbol. Due
to this, non-Markov processes cannot be represented by finite-state Markov
chains. One signature of (and motivation for) hidden Markov models is that
they can represent many non-Markov processes finitely. So, whenever a process
(Markov or not) has a finite-state
hidden Markov model presentation, then we say that the process is \emph{finite}.

There are a number of hidden Markov model variants. One common variant is a
state-emitting hidden Markov model~\cite{Baum1966a}. Another variant is an
edge-emitting hidden Markov model. State-emitting hidden Markov models output
symbols during state visitations, while edge-emitting hidden Markov models
output symbols on the transitions between states. The two variants are
equivalent~\cite{Uppe97a} in that they represent the same class of processes
finitely. In the following, we always refer to the edge-emitting variant.

As before, we restrict our attention to hidden Markov models whose underlying
Markov chain is irreducible. Thus, a hidden Markov model has a unique stationary
distribution $\pi$ satisfying $\pi = \pi \sum T_x = \pi T$ and, for $\alpha \in
\AlternateStateSet$, $\pi(\alpha)$ represents the stationary probability of
being in internal state $\alpha$.

For comparison, Fig.~\ref{fig:GMPHMM}(a) displays a hidden Markov model for the
Golden Mean Process. The internal state set is $\AlternateStateSet = \{A,B\}$
and the output alphabet is $\MeasAlphabet = \{0,1\}$. The transitions between
the states sport the labels $p|\Symbol{x}$, where $p = T_x(\alpha, \beta)$.

\begin{figure}
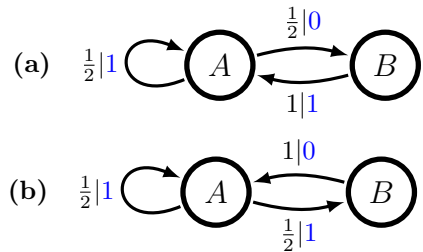

\centering

  \ifgenerate
    \ifcache
      \tikzsetnextfilename{gmp_hmm}
    \fi
    \input{\tikzpathgmp_hmm.tikz}
  \else
    \includegraphics{\tikzpathexternalgmp_hmm}
  \fi

  \ifgenerate
    \ifcache
      \tikzsetnextfilename{gmp_hmmtr}
    \fi
    \input{\tikzpathgmp_hmmtr.tikz}
  \else
    \includegraphics{\tikzpathexternalgmp_hmmtr}
  \fi

\caption{(a) The Golden Mean Process as a hidden Markov model. The internal
  state set is $\AlternateStateSet = \{A,B\}$ and the observation alphabet is
  $\MeasAlphabet = \{0,1\}$. The transitions between states specify
  $p|\Symbol{x}$ where $p = \Pr(X_0=x, R_1=\beta|R_0=\alpha)$.
  (b) Its time-reversed hidden Markov model is not the same. It
  is nonunifilar, while the forward presentation is.
  }
\label{fig:GMPHMM}
\end{figure}

The probability of any word is calculated as:
\begin{align*}
  &\hspace{-.2in}\Pr(X_0 = x_0,\ldots,X_{L-1}=x_{L-1})\\
  &=\sum_{\mathclap{\rho_0,\ldots,\rho_L}}
  \pi(\rho_0) T_{x_0}(\rho_0, \rho_1) \cdots
  T_{x_{L-1}}(\rho_{L-1}, \rho_{L}).
\end{align*}
In matrix form, with $T_w \equiv T_{w_0} \cdots T_{w_{L-1}}$, we have
\begin{align*}
  \Pr(X_{0:L}=w) = \pi T_w \one,
\end{align*}
where $\one = (1 1 \ldots 1 1)^t$.

The states $\AlternateStateSet$ and observations $\MeasAlphabet$ were synonymous
in Markov chains. The consequence of this was that every finite Markov process
was uniquely characterized by its transition matrix $T$. With hidden Markov
models, this is no longer true. A given process, even a Markov process, is
not uniquely characterized by a set of transition matrices $\{T_x\}$. To drive
this point home, Sec.~\ref{sec:structure} provides an example process that
has an uncountable number of presentations on a fixed, finite number of states.
This demonstrates the need for a canonical representation, which is also
introduced in Sec.~\ref{sec:structure}.

\subsection{Reversibility}
\label{sec:HMMReversibility}

In comparison to Markov chains, the literature on reversibility for hidden
Markov models is substantially smaller and not nearly as detailed---see,
for example, Ref.~\cite{MacDonald1997}.

Reversibility for Markov processes was defined, in Eq.~\eqref{eq:revself}, such
that the probability of every word equaled the probability of the reversed word.
We take this as a general definition, applicable even to non-Markov processes.
Thus, a process is \vocab{reversible} if and only if for all $w \in
\MeasAlphabet^L$ and all $L \in \mathbb{N}$, we have:
\begin{align}
  \Pr(X_{0:L} = w) = \Pr(X_{0:L} = \widetilde{w}),
\label{eq:NonMarkovProcessReversibility}
\end{align}
where, as before, $\widetilde{w}$ is the reversal of $w$.

Detailed balance plays a central role in Markov chains and their applications.
The analogous local-equilibrium property for hidden Markov models is more
subtle and interesting. We define \vocab{detailed balance for a hidden Markov
model} to mean that the following must hold for all $x\in\MeasAlphabet$ and
all $\alpha,\beta \in \AlternateStateSet$:
\begin{align}
  \pi(\alpha) T_x(\alpha,\beta) =
  \pi(\beta) T_x(\beta,\alpha)
  ~.
\end{align}

Trivially, if a hidden Markov model is in detailed balance, then its internal
Markov chain must also be in detailed balance. The converse, however, is not
true. Also, whenever a hidden Markov model is in detailed balance, one can show
that the process it generates is reversible. But unlike the Markov chain case,
detailed balance is \emph{not} equivalent to reversibility. And, quite generally,
the process generated by a hidden Markov model can be reversible even if the
model is not in detailed balance~\footnote{In Ref.~\cite{MacDonald1997}, it was
shown that Poisson-valued, state-emitting hidden Markov models are reversible
if their internal Markov chains are reversible. This result does not hold with
edge-emitting hidden Markov models, as demonstrated by example.}. The Golden
Mean Process of Fig.~\ref{fig:GMPHMM}(a) generates a reversible process, but
it is not in detailed balance. The contrapositive is perhaps more intriguing:
\emph{Every irreversible stationary process generated by a finite-state,
edge-emitting hidden Markov model is not in detailed balance}~\footnote{
A similar statement can be made of Markov chains since the process generated
by a Markov chain is reversible if and only the Markov chain is in
detailed balance.}.

We can use the condition of detailed balance to inspire a definition for the
time-reversed hidden Markov model. If $T_x$ are the labeled transition matrices
of a hidden Markov model and $\pi$ is its unique stationary distribution, then
its \vocab{time-reversed} hidden Markov model has labeled transition matrices
given by:
\begin{align}
  \widetilde{T}_x(\alpha,\beta)
    %&= \Pr(\widetilde{X}_0=x, \widetilde{R}_1=\beta \,|\, \widetilde{R}_0=\alpha) \nonumber \\
    &\equiv \Pr(X_0=x, R_0 = \beta \,|\, R_1=\alpha) \nonumber \\
    &=\frac{\pi(\beta) T_x(\beta,\alpha)}{\pi(\alpha)} ~. \label{eq:treM}
\end{align}
The time-reversed HMM for the Golden Mean Process is given in Fig.
\ref{fig:GMPHMM}(b), which is now nonunifilar.

As before, if $\pi$ is stationary for $T = \sum T_x$, then it is also
stationary for $\widetilde{T} = \sum \widetilde{T}_x$. To justify its
designation as the \emph{time-reversed} hidden Markov model, we demonstrate that
it does indeed generate words in reverse time and, thus, generates the
time-reversed process:
\begin{align*}
  &\hspace{-.2in}\Pr(\widetilde{X}_0 = x_0,\ldots,\widetilde{X}_{L-1}=x_{L-1})\\
  &=\sum_{\mathclap{\rho_0,\ldots,\rho_L}}
  \pi(\rho_0) \widetilde{T}_{x_0}(\rho_0, \rho_1) \cdots
  \widetilde{T}_{x_{L-1}}(\rho_{L-1}, \rho_{L}) \\
  &=\sum_{\mathclap{\rho_0,\ldots,\rho_L}}
  \pi(\rho_{L}) T_{x_{L-1}}(\rho_L, \rho_{L-1}) \cdots
  T_{x_{0}}(\rho_{1}, \rho_{0}) \\
    &= \Pr(X_0 = x_{L-1},\ldots,X_{L-1}=x_{0}) ~.
\end{align*}
This result provides an alternative characterization of reversibility which
parallels that for Markov chains given in Eq.~\eqref{eq:revcomp}. That is, a
hidden Markov model is reversible if and only if for all
$w \in \MeasAlphabet^L$ and all $L \in \mathbb{N}$, we have:
\begin{align}
  \label{eq:HMMReversibilityComparison}
  \Pr(X_{0:L} = w) = \Pr(\widetilde{X}_{0:L} = w),
\end{align}
indicating that the two hidden Markov models agree on the probability of every
word; cf. Eq.~(\ref{eq:NonMarkovProcessReversibility}). Also, note that if the
hidden Markov model is in detailed balance, then it equals the time-reversed
hidden Markov model: $T_x = \widetilde{T}_x$.
We see that detailed balance is a structurally restrictive property.

For Markov chains, determining if a process is reversible amounted to
checking for detailed balance. The situation is more complicated for hidden
Markov models but, curiously enough, there exists a straightforward procedure
to check if two hidden Markov models generate the same process language. This
is known as the \vocab{identifiability} problem~\cite{Blac57a}, and its
solution~\cite{Ito92a,Bala1993,Uppe97a}, though 20 years old now, does not seem
to be as well known. A crude test is to verify that the hidden Markov model and
its time-reversed hidden Markov model agree on the probabilities of every word
of length $L$, where $L < 2|\AlternateStateSet|$ and $|\AlternateStateSet|$ is
the number of states in the model~\cite{Uppe97a}.

Another interesting question is whether or not the reversibility of the internal
Markov chain has any effect on the reversibility of the observed process. As it
turns out, the answer is no. Jumping ahead a bit, we note that the forward
\eM\ in Fig.~\ref{fig:IrreversibleExample} has a reversible internal Markov
chain, but the observed process is irreversible. Additionally, to any
irreversible Markov chain, we can simply assign the same symbol on each
outgoing edge. This creates a period-$1$ process that is definitely reversible.
So, the reversibility of the internal Markov chain can make no statement on the
reversibility of the observed process.

\section{Structure and Canonical Presentations}
\label{sec:structure}

Rarely does one work directly with a process. Needless to say,
specifying the probability of every word at every length is a cumbersome
representation. Instead, one works with generators. However, one must be
careful in choosing a representation for the latter. For example, the class of
processes representable by finite-state hidden Markov models is strictly larger
than the class of processes representable by finite-state Markov
chains~\cite{Weis73}. So, one cannot use Markov chain presentations in many
cases.

As previously noted, when the process can be represented by a finite-state
Markov chain, then that presentation is unique. If the process has no
finite-state Markov chain representation, however, then there is a challenging
multiplicity of possible hidden Markov model presentations to choose from,
many with distinct structural properties. As an example,
Fig.~\ref{fig:NUGMPHMM} gives a continuously parametrized set hidden
Markov models for the Golden Mean Process. Each value of
$z=\Pr(B,0|A) \in [\frac{1}{2},1]$ defines a unique hidden Markov model that
generates the same Golden Mean Process. That is, $\Pr(X_1=\beta|X_0=\alpha)$ is
independent of $z$ and equal to the matrix $T(\alpha,\beta)$ that defined the
Markov chain in Fig.~\ref{fig:GMP}(a). Note that this is \emphasis{only} a
two-state hidden Markov model. It is possible to construct similar families
with even more states. (The technique for constructing such
continuously parametrized presentations for a given process will
appear elsewhere.)

\begin{figure}
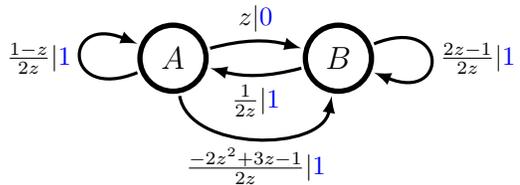

\centering

  \ifgenerate
    \ifcache
      \tikzsetnextfilename{gmp_continuous}
    \fi
    \input{\tikzpathgmp_continuous.tikz}
  \else
    \includegraphics{\tikzpathexternalgmp_continuous}
  \fi

\caption{The Golden Mean Process as a continuously parametrized hidden
  Markov model. The internal state set is $\AlternateStateSet = \{A,B\}$
  and the observation alphabet is $\MeasAlphabet = \{0,1\}$. Each value of
  $z = \Pr(B,0|A) \in [\frac{1}{2},1]$ defines a unique hidden Markov model
  that generates the same process as the models in Figs.~\ref{fig:GMP}(a)
  and \ref{fig:GMPHMM}(a).
  }
\label{fig:NUGMPHMM}
\end{figure}

This degeneracy serves to emphasize why a process's structure and that of its
presentations deserve close attention. To appreciate this concern more deeply,
we detour and examine structure explicitly. Then, we introduce \eMs\ and show
how they provide a canonical presentation that, in addition to other benefits,
resolves the degeneracy. Finally, we discuss additional notions of reversibility
that are more closely tied to and calculable from \eMs.

\subsection{Decomposing the State}
\label{sec:synccontrol}

Reference~\cite{Crut10a} presented an information-theoretic analysis of the
relationship between a hidden Markov model's states and the process it
generates. One of the main conclusions was that the internal-state uncertainty
$H[\AlternateState_0]$ can be decomposed into four independent components.
Here, we summarize the decomposition, assuming a minimal amount of information
theory. Reference~\cite{Cove06a} should be consulted for background not covered
here. Familiarity with the block entropy, entropy rate, and excess entropy as
developed in Ref.~\cite{Crut10a} is also assumed.

By splitting a process's bi-infinite sequence of random variables into a
past $X_{:0}$ and a future $X_{0:}$, we isolate the information that
passes through the \vocab{present} state $R_0$. As developed in
Refs.~\cite{Crut08a} and \cite{Crut08b}, the statistical relationships among
these three (aggregate) variables are concisely expressed using the
information diagram technique of Refs.~\cite{Yeun91a} and \cite{Yeun92a}.
Said briefly, a process's Shannon entropies and mutual informations~\cite{Shan62}
form a measure over the associated event (sequence) spaces. Given this, the
set-theoretic relationships between the measure's atoms are displayed in the
Venn-like diagram.

For a three-variable information diagram, we have three circles representing
$H[X_{:0}]$, $H[X_{0:}]$, and $H[R_{0}]$.  In total, this means that there are
$7$ atoms to consider. However, since every hidden Markov model has an internal
Markov chain that governs generation, the past and future are shielded from
each other given the current state. This is a probabilistic statement, but when
phrased in terms of conditional mutual information, we have
%\begin{align*}
  $I[X_{:0} ; X_{0:} | R_0] = 0$.
%\end{align*}
A moment's reflection shows that
this is a way of saying that the hidden Markov model generates
the process. This quantity can be nonzero only if we compare a process
to the states of a hidden Markov model that generates a
\emphasis{different} process.

\begin{figure}
\centering
\includegraphics[scale=.9]{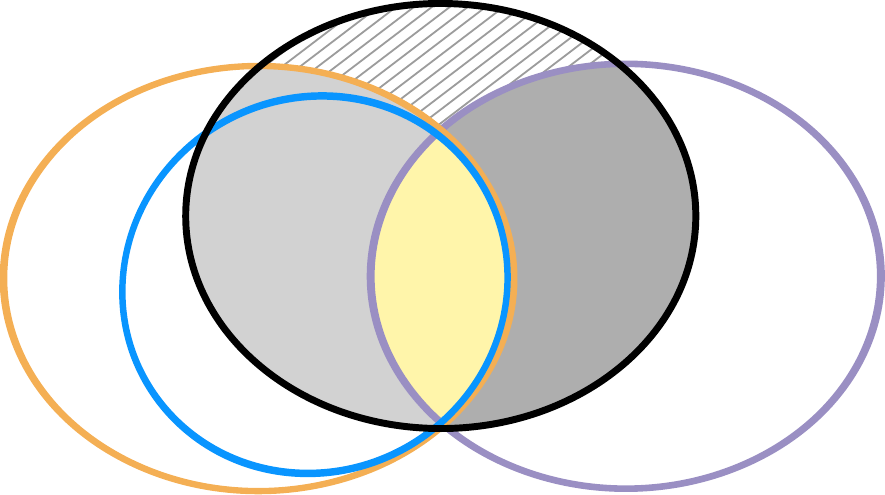}
\includegraphics[scale=.9]{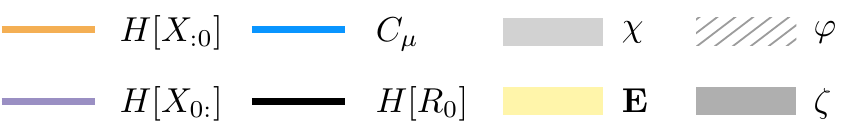}
\caption{Information diagram capturing all possible relationships between
  the past $X_{:0}$, the future $X_{0:}$, and the present---the current
  internal state $\AlternateState_0$. The statistical complexity
  $\Cmu$, excess entropy $\EE$, crypticity $\PC$, oracular
  information $\OI$, and gauge information $\GI$ appear as unions
  and intersections of the resulting atoms.
  }
\label{fig:synccontrol}
\end{figure}

The information diagram is shown in Fig.~\ref{fig:synccontrol}. There,
$H[X_{:0}]$ is represented by everything contained in the orange circle
\footnote{Keep in mind that, unless otherwise stated, these figures are not drawn to scale.
For example, the entropy of the past $H[X_{:0}]$ is infinite. Since
the drawings are not scale, we use the term \emphasis{circle} liberally.
Despite this, the important relationships of the variables are preserved.}.
The purple circle represents $H[X_{0:}]$ and the black circle, our focus,
represents \vocab{state information} $H[R_0]$. The figure contains an
additional blue circle that can be ignored until \eMs\ are introduced
in Sec.~\ref{sec:eMreview}. So, absent the blue circle, we see that the state
information decomposes into four quantities. Specifically,
\begin{equation}
H[\AlternateState_0] = \EE + \PC + \OI +\GI ~,
\end{equation}
where we have the:
\begin{enumerate}
\item \vocab{Excess entropy}:       $\EE = I[X_{:0}; X_{0:}]$,
\item \vocab{Crypticity}:           $\PC = I[X_{:0}; R_0 | X_{0:} ]$,
\item \vocab{Oracular information}: $\OI = I[R_0; X_{0:} | X_{:0} ]$, and
\item \vocab{Gauge information}:    $\GI = H[R_0 | X_{:0}, X_{0:} ]$.
\end{enumerate}

\vocab{Excess entropy} is a by-now standard measure of
complexity~\cite{Junc79,Crut83a,Gras86,Bial00a,Crut01a}
that captures the shared information between past and future observations.
\vocab{Crypticity} is a relatively new measure of structure introduced in
Refs.~\cite{Crut08a, Crut08b, Maho09a}. By comparing to the apparent
information that excess entropy measures, crypticity monitors how much of
the internal state information is hidden.
\vocab{Oracular information}, introduced in Ref.~\cite{Crut10a},
measures how much information a presentation provides that can improve
predictability, but that is not available from the past.
Finally, \vocab{gauge information}, also introduced in Ref.~\cite{Crut10a},
quantifies how much additional structural information exists in a
presentation that is not ``justified'' by the past or the future.
Taken together these quantities provide an informational basis
useful for analyzing the various kinds of structure a process or a
process's presentation contains.

To see this, we can apply these structural complexity measures to the Golden
Mean Process presentation family of Fig.~\ref{fig:NUGMPHMM}. For each value of
$z = \Pr(B,0|A) \in [\frac{1}{2}, 1]$, Fig.~\ref{fig:gmpsd} plots $\EE$, $\PC$,
$\OI$, and $\GI$ stacked in way so that their sum $H[\AlternateState_0]$ is
the top curve. One immediately sees that $\EE$ is independent of $z$. This is
as it should be since $\EE$ is a function only of the observed process and,
by construction, the parametrized presentation always generates the Golden
Mean Process. All of the other measures change as the presentation changes,
however. Let's explore what they tell us.

\begin{figure}
\centering
\hspace{.2in}\includegraphics[scale=.28]{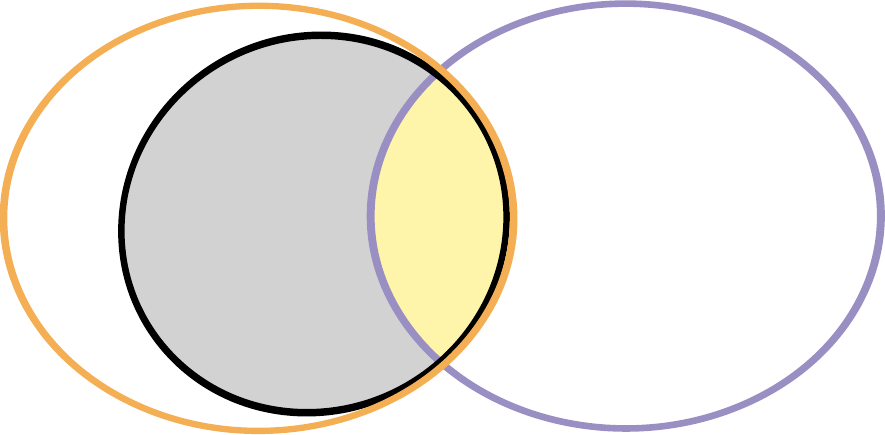}
\includegraphics[scale=.28]{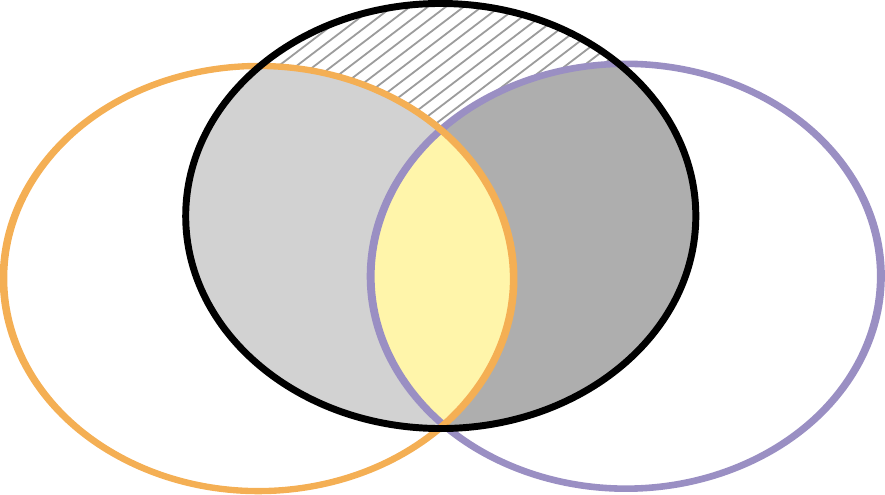}
\includegraphics[scale=.28]{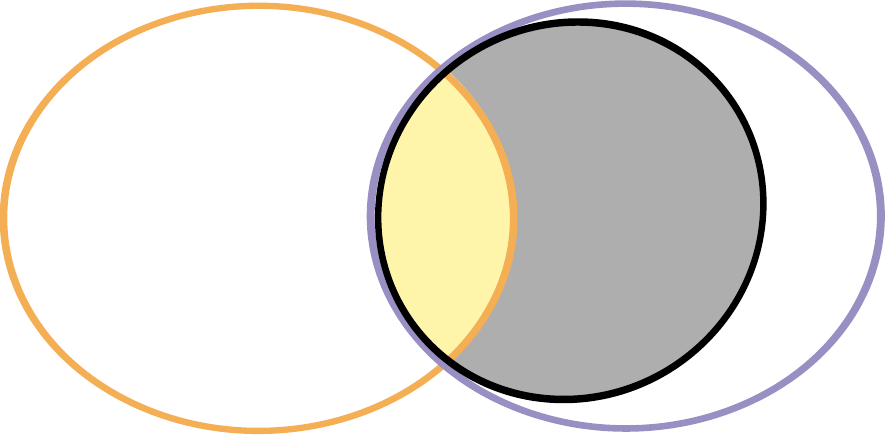}
\includegraphics[width=\columnwidth]{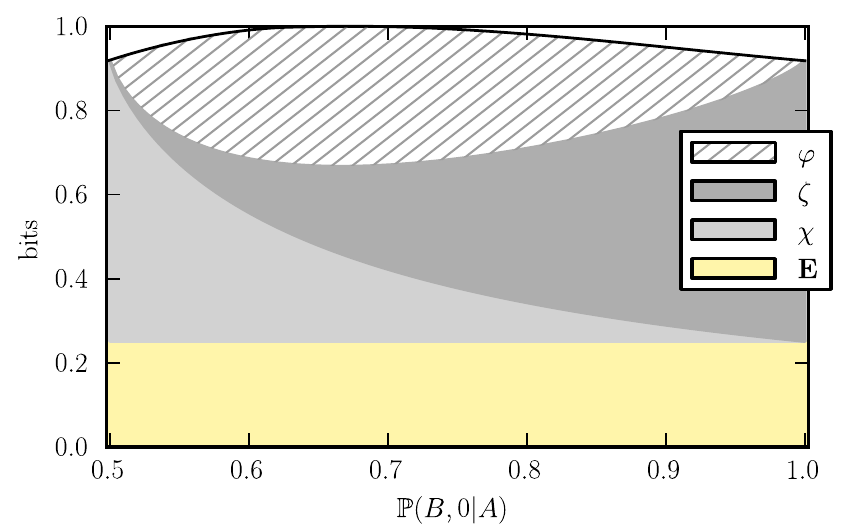}
\caption{Decomposition of the state information $H[\AlternateState]$ contained
  in the parametrized Golden Mean Process presentation family of
  Fig.~\ref{fig:NUGMPHMM}. As a function of $z = \Pr(B,0|A) \in [\frac{1}{2},
  1]$, the excess entropy $\EE$, crypticity $\PC$, oracular information $\OI$,
  and gauge information $\GI$ are stacked such that the top of the curve is
  their sum $H[\AlternateState]$, the state entropy of the presentation for the
  given value of $z$. The miniaturized information diagrams are special cases of
  Fig.~\ref{fig:synccontrol} tailored to $z$-values. From left-to-right, we have
  $z=\frac{1}{2}$, $z = \frac{3}{4}$, and $z=1$. }
\label{fig:gmpsd}
\end{figure}

Beginning with $z=1/2$, we recover the Markov chain presentation of
Fig.~\ref{fig:GMPHMM}(a). In this presentation, all of the state information
$H[\AlternateState_0]$ is contained within $H[X_{:0}]$. This is represented by
the leftmost information diagram at the top of Fig.~\ref{fig:gmpsd}. Loosely,
we say that the state
information contains only information from the past. However, one must keep in
mind that the presentation still captures $\EE$ bits of information, and this
information is shared with the future. The gauge and oracular informations
vanish. It turns out that the $z=1/2$ presentation is the process's forward
\eM, but more on this later.

As $z$ increases, so do the gauge and oracular informations. With
this change, the information diagram circle for $H[\AlternateState_0]$ straddles
$H[X_{:0}]$ and $H[X_{0:}]$, as shown in the central information diagram
atop Fig.~\ref{fig:gmpsd}. This indicates that the state information now
consists of historical information, oracular information, and also gauge
information. For all values of $z$, the overlap that $H[\AlternateState_0]$ has
with the intersection of the past and future is constant. This is because each
presentation generates the process and so each must capture $\EE$ bits of
shared information.

Finally when $z = 1$, the circle for $H[\AlternateState_0]$ is now completely
contained inside the future $H[X_{0:}]$. Now, the information diagram
resembles the right-most one atop Fig.~\ref{fig:gmpsd}. There is no crypticity,
no gauge information, but there is oracular information. The interpretation is
that the state information, apart from $\EE$, consists only of information from
the future. As we will see, the $z=1$ presentation corresponds to the
time-reversed presentation of the reverse \eM. And, since the Golden Mean
Process is a reversible Markov chain, the $z=1$ information diagram mirrors
the diagram for $z=\frac{1}{2}$.

\subsection{\texorpdfstring{\EMs}{eMachines}}
\label{sec:eMreview}

We discussed processes in the context of generators, as represented by Markov
chains and hidden Markov models, but another important aspect concerns prediction.
As we will show, \eMs\ are a natural consequence of this perspective, and
they provide a much richer analysis of irreversibility.  Additionally, their
uniqueness provides a solution to the multiplicity of HMM presentations.

Consider again a process's output sequence and, now, interpret time as
increasing with the index. The result is a time-series
$\ldots \MeasSymbol_{t-1} \MeasSymbol_t \MeasSymbol_{t+1} \ldots$. Our goal
is to construct a model that predicts future observations. Specifically,
we want to find sufficient statistics that preserve our ability to predict.
Translating this into a concrete procedure, we first remove redundancies in
the time-series, by grouping histories that lead to the same distribution
over futures:
\begin{align*}
    \meassymbol_{:0} \sim \meassymbol_{:0}^\prime
    \iff
    \Pr(X_{0:} | X_{:0} = \meassymbol_{:0})
    = \Pr(X_{0:} | X_{:0} = \meassymbol_{:0}^\prime) ~.
\end{align*}
The grouping defines an equivalence relation over histories and, thus,
partitions the space of histories. This partition is the coarsest one that
provides optimal prediction. It is called the process's
\vocab{causal state partition}. Each equivalence class is known as a
\vocab{causal state} and, thus, to each causal state, there is a unique
distribution over futures \cite{Crut88a,Crut92c,Shal98a}.
The set of causal states is denoted $\CausalStateSet$.

Now, consider a semi-infinite history $X_{:0} = x_{:0}$ which,
by the causal state equivalence relation, induces causal state
$\CausalState_0=\sigma_0$.  If we append a new observation, we get
$X_{:1} = x_{:0} x_0$ which, in turn, induces $\CausalState_1=\sigma_1$.
In this sense, there is a natural dynamic over the causal states that is
induced by the dynamic over the observed sequences. This dynamic is represented
in Fig.~\ref{fig:eMdynamic}. The pair of causal states and transition dynamic
is called a process's \vocab{\eM}.

\begin{figure}
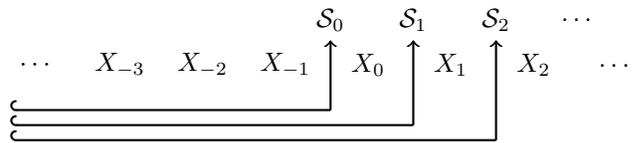

\centering

  \ifgenerate
    \ifcache
      \tikzsetnextfilename{eMdynamic}
    \fi
    \input{\tikzpatheMdynamic.tikz}
  \else
    \includegraphics{\tikzpathexternaleMdynamic}
  \fi

\caption{The dynamic over the causal states is induced by the dynamic
  over the (semi-infinite) histories.  For example, a history
  $\MeasSymbol_{-1:}$ ending at $t=-1$, maps to causal state
  $\CausalState_0$.  When a new symbol $X_0$ is appended to the old history,
  we induce a new causal state $\CausalState_1$.
  }
\label{fig:eMdynamic}
\end{figure}

Generally, the set of causal states can be uncountable, countable, or finite;
see, for example, Fig. 17 in Ref. \cite{Crut92c}.
Even when the set is not finite, the set that is visited infinitely often
may be finite. The infinitely visited subset defines the \vocab{recurrent
causal states}. All other states are \vocab{transient causal states} and
not the subject of our discussion here.  Now, when the set of recurrent causal
states is finite, then the \eM---obtained by partitioning histories for the
purposes of prediction---is representable as a finite-state unifilar
hidden Markov model.  We denote the transition matrices in the same way,
except that we use $\CausalState$ as the state random variables, which take on values
from $\CausalStateSet$:
\begin{align*}
  T_x(\alpha,\beta) =
  \Pr(\MeasSymbol_0=x, \CausalState_1=\beta | \CausalState_0=\alpha)~.
\end{align*}
\EMs\ with a finite number of recurrent states generate a subset of the
\vocab{finitary} processes---processes with finite excess entropy.
This subset represents a strictly larger set of processes than
finite-state Markov chains since it includes processes with measures over
strictly sofic~\cite{Cove75} shifts.

The \eM\ is the unique presentation in the class of unifilar hidden Markov
models~\cite{Crut88a,Shal98a} and, thus, it defines a canonical presentation
for a given process. There are other benefits. For one, \eM\
unifilarity allows one to directly calculate the process entropy rate. Early
on, Shannon pointed out that this is always possible to do with Markov
chains. It was soon discovered that it is not possible using nonunifilar
hidden Markov models \cite{Blac57a}. Nonunifilarity makes each presentation
state appear more random than it actually is. For a more detailed treatment of
\eMs, see Ref.~\cite{Crut88a}.

We pause briefly to point out that the unifilarity property of the \eM\ is a
consequence of the equivalence relation. It has been known for some
time~\cite{Crut92c,Crut08b,Lohr2009a} that there are nonunifilar hidden
Markov models of processes that can be smaller, sometimes substantially
smaller, than the process's \eM.  However, finding a canonical presentation
within the class of nonunifilar hidden Markov models is a task that has evaded
solution.  One obvious choice is to focus on the hidden Markov model that
minimizes the state entropy; see Ref.~\cite{Lohr2009a} for further discussion.
Since our goal is to analyze the role that structure plays in irreversibility,
having a canonical representation is essential. So, our focus on \eMs\ is based,
in part, on practicality since one can \emph{calculate} the \eM\
from any alternative presentation. It is also theoretically useful since
many quantities---such as the process entropy rate and excess entropy---are
not exactly calculable from nonunifilar presentations.  Additionally, the
states of nonunifilar presentations are not sufficient statistics for the
histories. The consequence is that one cannot forget the past and work with
an individual state in a general hidden Markov model---instead, one must work
with a distribution over the states~\footnote{The entropy of the distribution of
distributions over states is precisely the \protect\eM's statistical
complexity $C_\mu = H[\CausalState]$.}.

Our discussion of processes began by pointing out that \emph{time} is merely
an interpretation of the indices on a set of random variables. Thus far, we
described \eMs\ from the forward perspective---yielding the forward \eM,
denoted $\FutureEM$. Similarly, following Refs. \cite{Crut08a,Crut08b} one can
partition futures for the purposes of retrodiction, and this partitioning
induces a dynamic over the reverse causal states. The resulting unifilar hidden
Markov model is known as the reverse \eM, denoted $\ReverseEM$. To
differentiate the states in each hidden Markov model, we let
$\ForwardCausalState_t$ represent the random variables for the
\vocab{forward causal states} and use $\ReverseCausalState_t$ for the
\vocab{reverse causal states}.  The equivalence relations used during
partitioning, $\sim^+$ and $\sim^-$, are generally distinct. We use
$\epsilon^+ : \past \to \FutureCausalState$ to denote
the mapping that takes a history and returns the forward causal state into
which the history was partitioned. Similarly, we use $\epsilon^- :
\future \to \PastCausalState$ to denote the
mapping from futures to reverse causal states.

\begin{figure}
\centering

  \ifgenerate
    \ifcache
      \tikzsetnextfilename{processlattice}
    \fi
    \input{\tikzpathprocesslattice.tikz}
  \else
    \includegraphics{\tikzpathexternalprocesslattice}
  \fi

\caption{Hidden Process Lattice: The $\MeasSymbol$ variables denote the
  observed process; the $\CausalState$ variables, the hidden causal states.
  If one scans the observed variables in the positive direction---seeing
  $\MeasSymbol_{-3}$, $\MeasSymbol_{-2}$, and $\MeasSymbol_{-1}$---then that
  history takes one to causal state $\ForwardCausalState_0$. Analogously, if
  one scans in the reverse direction, then the succession of variables
  $\MeasSymbol_{2}$, $\MeasSymbol_{1}$, and $\MeasSymbol_{0}$ leads to
  $\ReverseCausalState_0$. The colors indicate which variables participate
  in the information measures of Fig.~\ref{fig:eMidiagram}.
  }
\label{fig:ProcessLattice}
\end{figure}

To orient ourselves, Fig.~\ref{fig:ProcessLattice} places the relevant random
variables on a lattice. The $\MeasSymbol$ variables denote the observed process
of the hidden Markov model, which is broken up into the past (orange) and
future (purple) observation sequences.  The hidden causal states are
represented by the $\CausalState$ variables.  In the \emph{present}, we
have $\ForwardCausalState_0$ and $\PastCausalState_0$ straddling the
\emph{past} and \emph{future}. If one scans the observed variables in the
positive direction---seeing $\MeasSymbol_{-3}$, $\MeasSymbol_{-2}$,
and $\MeasSymbol_{-1}$---then that history takes one to causal state
$\ForwardCausalState_0$. Analogously, if one scans in the reverse direction,
then the succession of variables $\MeasSymbol_{2}$, $\MeasSymbol_{1}$, and
$\MeasSymbol_{0}$ leads to $\ReverseCausalState_0$.

Summarizing, we represent each \eM\ as a commuting diagram that operates
on the hidden process lattice, using $x$ and $\sigma$ to represent symbol and
causal state realizations, respectively:
\begin{center}

  \ifgenerate
    \ifcache
      \tikzsetnextfilename{commutingdiagram_feM}
    \fi
    \input{\tikzpathcommutingdiagram_feM.tikz}
  \else
    \includegraphics{\tikzpathexternalcommutingdiagram_feM}
  \fi

\qquad

  \ifgenerate
    \ifcache
      \tikzsetnextfilename{commutingdiagram_reM}
    \fi
    \input{\tikzpathcommutingdiagram_reM.tikz}
  \else
    \includegraphics{\tikzpathexternalcommutingdiagram_reM}
  \fi

\end{center}
For the forward \eM\ $M^+$, every past $\meassymbol_{:0}$ maps to a unique next
past $\meassymbol_{:1}$ on symbol $\meassymbol_0$.  By the forward-looking
map $\epsilon^+$, each past $\meassymbol_{:0}$ corresponds to
unique causal state $\forwardcausalstate_0$. This many-to-one correspondence
induces a dynamic on the causal states such that $\futurecausalstate_0$
transitions to $\forwardcausalstate_1$ on symbol $\meassymbol_0$.  Similarly,
for the reverse \eM\ $\ReverseEM$, every future $\meassymbol_{1:}$ maps to a
unique next future $\meassymbol_{0:}$ on symbol $\meassymbol_0$.  The
reverse-looking map $\epsilon^-$ associates $\meassymbol_{1:}$ with
$\causalstate^-_1$. The many-to-one correspondence induces a dynamic on the
reverse causal states such that $\reversecausalstate_1$ transitions to
$\reversecausalstate_0$ on symbol $\meassymbol_0$.

\begin{figure}
\includegraphics[scale=.7]{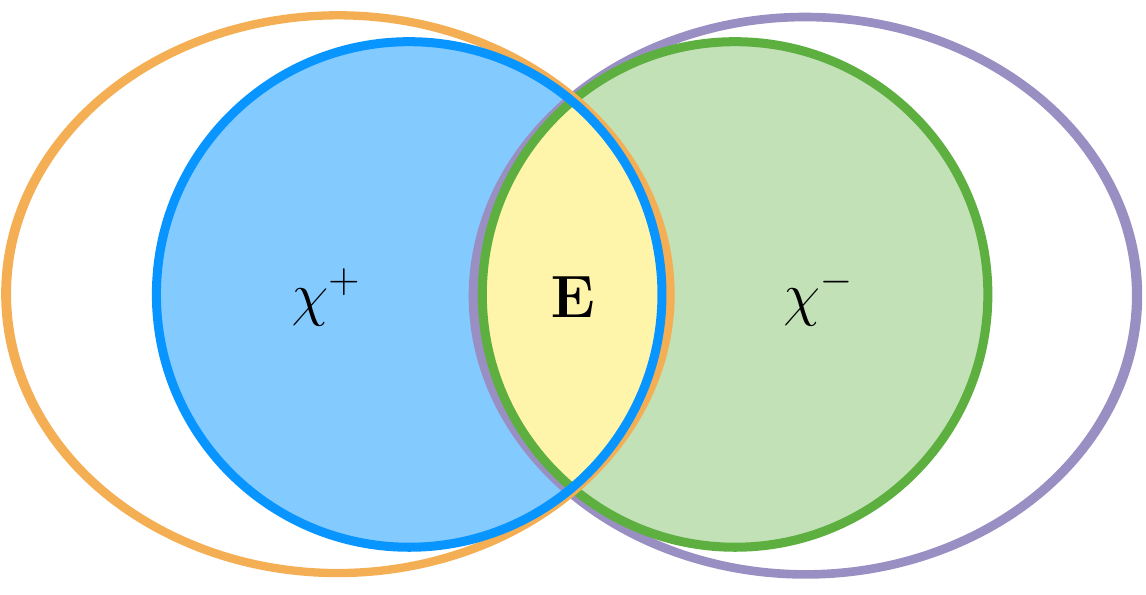}\\
\vspace{.1in}
%\hspace{.5in} % centers lines, not legend
\includegraphics[scale=1]{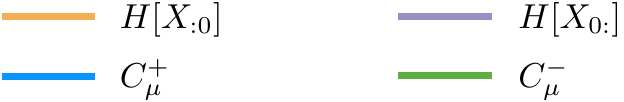}
\caption{Information diagram for the forward and reverse \eMs.}
\label{fig:eMidiagram}
\end{figure}

Finally, we gather the forward and reverse \eMs\ in Fig.~\ref{fig:eMidiagram}.
Together, they provide complementary views of the process. For example, the
minimal amount of information one must store in order to generate the process
in the forward direction defines the \vocab{forward statistical complexity}
$\FutureCmu \equiv H[\ForwardCausalState]$. This information, in general,
is not equal to the minimal amount of information one requires for retrodiction
$\PastCmu \equiv
H[\PastCausalState]$~\cite{Crut08a,Crut08b,Maho09a}. Notably,
the \eM\ has no gauge information since it is minimal and, also, no oracular
information since it is unifilar. Referring briefly back to
Figs.~\ref{fig:synccontrol} and \ref{fig:gmpsd}, when $z=\frac{1}{2}$, we have
the forward \eM.  When $z=1$, we have the time-reversed presentation of the
reverse \eM\ $\widetilde{M}^-$. The interpretation is direct: The crypticity
$\PC^-$ of the reverse \eM\ becomes oracular information $\OI$ in the
time-reversed presentation.

Our preference, from now on, is to use the forward and reverse \eMs.
Given the forward \eM\ $M^+$, we can construct, via Eq.~\eqref{eq:treM}, a
reverse generator of the process $\widetilde{M}^+$. However, this is just one
presentation among many possible reverse generators of the process.  So, we
operate on that reverse generator, using techniques from Ref.~\cite{Crut08b},
and obtain the reverse \eM\ $M^-$.  Together, the forward and reverse \eMs\
serve as the basis for understanding processes through the use of generators.
In Sec.~\ref{sec:bieM}, we unify the two \eMs\ into a single machine and
discuss its meaning in the context of the decomposition of state information.

\subsection{Finite State Automata}

One interesting property of \eMs, and hidden Markov models in general, is that
they are intimately related to automata in formal language theory \cite{Hopc79}.
Here, we briefly review their relationship.

Given a process $\Process$, we can examine the set of all words that occur
with positive probability. This set is known as the \vocab{support} of the
process's stochastic language.  Stripping away the transition probabilities of
any finite-state hidden Markov model leaves
a finite-state automaton that generates the support of the process language.
So, we see that the support of a process generated by finite-state HMM always
corresponds to a \emph{regular} language.  If the hidden Markov model was
unifilar, then the resulting structure, without probabilities, is equivalent
to a deterministic finite automata (DFA). Similarly, nonunifilar hidden
Markov models map to nondeterministic finite automata (NFA).

However, it is necessary to point out that there are quite drastic differences
between formal and process languages. While DFAs and NFAs are equivalent in
the set of formal languages that each can represent using a finite number of
states, the same is not true of hidden Markov models. In fact, there are
finite-state nonunifilar HMMs that have no corresponding finite-state unifilar
counterpart.  One well known example is the
\vocab{Simple Nondeterministic Source} of Ref.~\cite{Crut92c}. It can be
represented as a two-state nonunifilar HMM, but its \eM---the smallest unifilar
HMM generating the same process---requires a countably infinite number of states.

Since it will be useful to compare topological properties to statistical
properties, we define $M_\emptyset^+$ and $M_\emptyset^-$ as the deterministic
finite-state automata corresponding to the forward and reverse \eMs\ with
all probabilities removed. Note that these DFAs need not be the minimal
deterministic finite-state automata~\cite{Hopc79} generating the support, and
this fact highlights the difference between the causal-state equivalence
relation and the Nerode state-equivalence relation in formal language theory.
If we
subsequently minimize $M_\emptyset^+$ and $M_\emptyset^-$, we are left with the
minimal and unique DFAs that generate the support, respectively denoted
$\ForwardDFA$ and $\ReverseDFA$.

Also, we mention that there is a large body of literature in formal language
theory concerning $k$-reversible
languages~\cite{Angluin1982,Pin1992,Lombardy2002,Sempere2006a,Garcia2009}.
This topic does not relate directly to our notion of reversibility and is
rather closer to addressing a process's Markov order; cf. Ref. \cite{Jame10a}.

One can view \eMs\ as probabilistic counterparts to DFAs.
In fact, the relation between formal language theory and stochastic languages
can be extended. Just as there is a hierarchy of models in formal
language theory, one can consider a hierarchy of stochastic models as well.
See, for example, the process hierarchy proposed in Ref.~\cite{Crut92c}.

\subsection{Reversibility Revisited}

By focusing on \eMs, we side-step the representational degeneracy of hidden
Markov models. Recalling the uniqueness of the forward and reverse \eMs,
we note that properties of the \eM\ can also be interpreted as properties
of the process. This also allows us to consider additional measures of
reversibility that are based on structural properties of the \eM.
So, while each of the forthcoming definitions can be stated
strictly in terms of the process's probability distribution, we prefer to use
equivalent definitions in terms of the forward and reverse \eMs.  This
is akin to studying formal languages through the use of the minimal DFAs.

As Ref.~\cite{Uppe97a} demonstrated, there is a finite procedure for determining
whether two finite-state hidden Markov models generate the same process language.
By Eq.~\eqref{eq:HMMReversibilityComparison}, this technique also provides a
method for determining whether a process is reversible or not. An alternate
technique involves the forward and reverse \eMs. With them, one simply asks
if the two machines are identical to each other. If so, then the process is
reversible. In Ref.~\cite{Crut08b}, this property was termed
\emph{microscopic reversibility} and we write: $M^+ = M^-$.

We can also consider several weaker forms of reversibility. For example, as we noted,
the process that repeats $ABC\ldots$ indefinitely is not reversible, but the
\eMs\ are essentially the same in that the amount of information one requires
for prediction equals the amount required for retrodiction. Following
Ref.~\cite{Crut08b}, a process is \emph{causally reversible} if and only if
$\ForwardCmu = \ReverseCmu$.

In terms of topology, we say that a process is \emph{support reversible}
if and only if $\ForwardDFA = \ReverseDFA$, where equality means that DFAs
must be identical under an isomorphism over the states. Finally, we also
consider symbol isomorphisms.  If there exists an isomorphism
from the output alphabet of $\ForwardEM$ to the output alphabet of $\ReverseEM$
that renders the two machines equal, then we say that the process is
\emph{reversible under symbol isomorphism}, denoted
$\ForwardEM \cong \ReverseEM$. Similarly, the process is
\emph{support reversible under symbol isomorphism} if and only if
$\ForwardDFA \cong \ReverseDFA$.

\section{Examples}

This section exercises the preceding theory, giving a number of additional
results and illustrating them through example processes and presentations. We
start with an exploration of which kinds of reversibility there can be. Then we
analyze in detail two example irreversible processes, one with a rather
counterintuitive property. The analyses give a concrete understanding of how
irreversibility arises and what its structural consequences are for a process.
The section closes with a survey that demonstrates the dominance of
irreversibility among processes.

\subsection{Causal Reversibility Roadmap}
\label{sec:sec:diversitymatrix}

Given these various notions of reversibility, a natural question comes to mind:
What combinations are possible? To this end, we state a number of
straightforward relationships:
\begin{align}
    \ForwardEM = \ReverseEM
      \quad &\Rightarrow \quad
    \FutureCmu = \PastCmu
	\label{eq:EMEquivCausalRev} \\
    \ForwardEM = \ReverseEM
      \quad &\Rightarrow \quad
    \ForwardDFA = \ReverseDFA\\
    \ForwardEM = \ReverseEM
      \quad &\Rightarrow \quad
    \ForwardEM \cong \ReverseEM\\
    \ForwardEM \cong \ReverseEM
      \quad &\Rightarrow \quad
    \ForwardDFA \cong \ReverseDFA\\
    \ForwardDFA = \ReverseDFA
      \quad &\Rightarrow \quad
    \ForwardDFA \cong \ReverseDFA
	\label{eq:DFAEquivDFAIso}
\end{align}

Now, let us restrict attention to just the causally reversible processes
($\FutureCmu = \PastCmu$) and examine microscopic and support reversibility,
with and without symbol isomorphisms. That is, we consider the combinations of
the four properties (i) $\ForwardEM = \ReverseEM$, (ii)
$\ForwardEM \cong \ReverseEM$, (iii)
$\ForwardDFA = \ReverseDFA$, and (iv) $\ForwardDFA \cong \ReverseDFA$.
Of the $16$ possible Boolean-vector combinations, only $6$ are possible due
Eqs. (\ref{eq:EMEquivCausalRev}) - (\ref{eq:DFAEquivDFAIso}).

%%% AUTOMATED INPUT %%%
%\input{diversitymatrix}

% http://www.tex.ac.uk/cgi-bin/texfaq2html?label=tabcellalign
\newcolumntype{M}{ >{\centering\arraybackslash} m{4cm} }
\newcolumntype{N}{ >{\centering\arraybackslash} m{2cm} }

\begin{table*}
\centering
\begin{tabular}{|M|M|N|N|N|N|}
  \hline
  $M^+ \vphantom{M^{M^M}_{M_M}}$ & 
  $M^- \vphantom{M^{M^M}_{M_M}}$ &
  $\ForwardEM = \ReverseEM \vphantom{M^{M^M}_{M_M}}$ &
  $\ForwardEM \cong \ReverseEM \vphantom{M^{M^M}_{M_M}}$ &
  $\ForwardDFA = \ReverseDFA \vphantom{M^{M^M}_{M_M}}$ &
  $\ForwardDFA \cong \ReverseDFA \vphantom{M^{M^M}_{M_M}}$ \\
  \hline\hline
%
%%%
%%% F F F F
%%% ICDFA.int_to_machine(16,2,3)
%%%
  
  \ifgenerate
    \ifcache
      \tikzsetnextfilename{diversity_FFFF_feM}
    \fi
    \input{\tikzpathdiversity_FFFF_feM.tikz}
  \else
    \includegraphics{\tikzpathexternaldiversity_FFFF_feM}
  \fi

  &
  
  \ifgenerate
    \ifcache
      \tikzsetnextfilename{diversity_FFFF_reM}
    \fi
    \input{\tikzpathdiversity_FFFF_reM.tikz}
  \else
    \includegraphics{\tikzpathexternaldiversity_FFFF_reM}
  \fi

  &
  F & F & F & F\\\hline 
%
%%%
%%% F F F T
%%% ICDFA.int_to_machine(97,3,3)
%%%
  
  \ifgenerate
    \ifcache
      \tikzsetnextfilename{diversity_FFFT_feM}
    \fi
    \input{\tikzpathdiversity_FFFT_feM.tikz}
  \else
    \includegraphics{\tikzpathexternaldiversity_FFFT_feM}
  \fi

  &
  
  \ifgenerate
    \ifcache
      \tikzsetnextfilename{diversity_FFFT_reM}
    \fi
    \input{\tikzpathdiversity_FFFT_reM.tikz}
  \else
    \includegraphics{\tikzpathexternaldiversity_FFFT_reM}
  \fi

  &
  F & F & F & T\\\hline
%
%%%
%%% F F T F
%%% Impossible since D+ = D- implies D+ \cong D-
%%%
%
%%%
%%% F F T T
%%% ICDFA.int_to_machine(43,2,3)
%%% 
  
  \ifgenerate
    \ifcache
      \tikzsetnextfilename{diversity_FFTT_feM}
    \fi
    \input{\tikzpathdiversity_FFTT_feM.tikz}
  \else
    \includegraphics{\tikzpathexternaldiversity_FFTT_feM}
  \fi

  &
  
  \ifgenerate
    \ifcache
      \tikzsetnextfilename{diversity_FFTT_reM}
    \fi
    \input{\tikzpathdiversity_FFTT_reM.tikz}
  \else
    \includegraphics{\tikzpathexternaldiversity_FFTT_reM}
  \fi

  &
  F & F & T & T\\\hline
%
%%%
%%% F T F F 
%%% Impossible since M+ \cong M- implies D+ \cong D-
%%%
%
%%%
%%% F T F T
%%% ICDFA.int_to_machine(15,2,3)
%%%
  
  \ifgenerate
    \ifcache
      \tikzsetnextfilename{diversity_FTFT_feM}
    \fi
    \input{\tikzpathdiversity_FTFT_feM.tikz}
  \else
    \includegraphics{\tikzpathexternaldiversity_FTFT_feM}
  \fi

  &
  
  \ifgenerate
    \ifcache
      \tikzsetnextfilename{diversity_FTFT_reM}
    \fi
    \input{\tikzpathdiversity_FTFT_reM.tikz}
  \else
    \includegraphics{\tikzpathexternaldiversity_FTFT_reM}
  \fi

  &
  F & T & F & T\\\hline
%
%%%
%%% F T T F
%%% Impossible since D+ = D- implies D+ \cong D-
%%%
%
%%%
%%% F T T T
%%% Not a topological eM.
%%%
  
  \ifgenerate
    \ifcache
      \tikzsetnextfilename{diversity_FTTT_feM}
    \fi
    \input{\tikzpathdiversity_FTTT_feM.tikz}
  \else
    \includegraphics{\tikzpathexternaldiversity_FTTT_feM}
  \fi

  &
  
  \ifgenerate
    \ifcache
      \tikzsetnextfilename{diversity_FTTT_reM}
    \fi
    \input{\tikzpathdiversity_FTTT_reM.tikz}
  \else
    \includegraphics{\tikzpathexternaldiversity_FTTT_reM}
  \fi

  &
  F & T & T & T\\\hline
%
%  & & T & F & F & F\\
%  & & T & F & F & T\\
%  & & T & F & T & F\\
%  & & T & F & T & T\\
%  & & T & T & F & F\\
%  & & T & T & F & T\\
%  & & T & T & T & F\\
%%%
%%% T T T T
%%% ICDFA.int_to_machine(7,2,2)
%%%
  
  \ifgenerate
    \ifcache
      \tikzsetnextfilename{diversity_TTTT_feM}
    \fi
    \input{\tikzpathdiversity_TTTT_feM.tikz}
  \else
    \includegraphics{\tikzpathexternaldiversity_TTTT_feM}
  \fi

  &
  
  \ifgenerate
    \ifcache
      \tikzsetnextfilename{diversity_TTTT_reM}
    \fi
    \input{\tikzpathdiversity_TTTT_reM.tikz}
  \else
    \includegraphics{\tikzpathexternaldiversity_TTTT_reM}
  \fi

  &
  T & T & T & T\\
  \hline
\end{tabular}
\caption{Diversity of causally reversible processes ($\Cmu^+ = \Cmu^-$):
  Example presentations for forward and reverse \eM\ pairs, with the same
  number of states, for the $6$ possible combinations; all other
  combinations are impossible.
  }
\label{table:diversitymatrix}
\end{table*}

Table~\ref{table:diversitymatrix} gives example forward and reverse \eM\ pairs
for each of the $6$ possibilities.  What we learn from these examples is that
causal reversibility indeed captures a larger class of processes than
microscopic reversibility. However, it also captures a bit more, including
processes that are not isomorphic to one another under a symbol isomorphism.
The table also demonstrates that irreversibility is not \emph{only}
a topological concern---the forward and reverse DFAs can be identical while
the generated process languages are not.

\subsection{Causal Irreversibility}

Irreversible processes are ubiquitous, even among those represented by
finite-state \eMs. In our first example, we ground intuitions with a
process whose irreversibility is driven topologically. The example
is particularly illustrative since its \eMs\ have a finite number of
causal states. In the second example, we examine an irreversible process
whose forward and reverse DFAs are identical; this demonstrates that
irreversibility can arise purely probabilistically. Then, in the third example,
we see the extent to which probability aggravates irreversibility when it
causes a finite-state forward \eM\ to become an infinite-state reverse \eM.

\subsubsection{Support-Driven Irreversibility}

The first example we consider shows that a process can have different, but
finite, numbers of forward and reverse causal states.
Formally, Ref.~\cite{Crut08b} provides the technique for calculating the
reverse \eM\ via operations on the graph structure of the forward
\eM\ but, for pedagogical reasons, both the forward and reverse causal states
are constructed in terms of $X_t$ only~\footnote{In this work and also in
Ref.~\cite{Crut08b}, two equivalence relations were defined. The forward
equivalence relation $\sim^+$ partitioned $X_{:0}$, while the reverse
equivalence relation $\sim^-$ partitioned $X_{0:}$.  However, these
relations are formally the same in that they both partition a generator's
\emph{local time} histories. To see this, recall that $X_{0:} / \sim^-$ is
isomorphic to $\protect\widetilde{X}_{:0} / \sim^+$.}.

\begin{figure*}
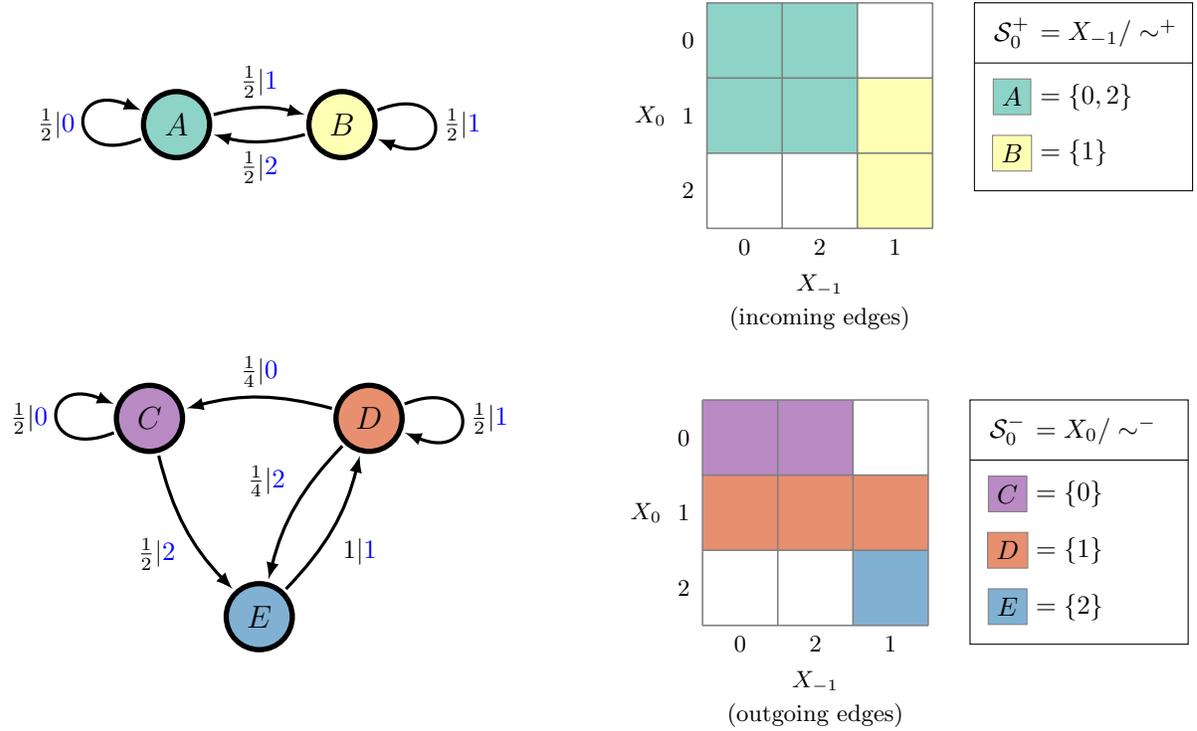

\noindent%
\begin{minipage}{\columnwidth}
  \centering
  \vspace{-.15in}
  
  \ifgenerate
    \ifcache
      \tikzsetnextfilename{irrev2to3_feM}
    \fi
    \input{\tikzpathirrev2to3_feM.tikz}
  \else
    \includegraphics{\tikzpathexternalirrev2to3_feM}
  \fi

  \\\vspace{.8in}
  \noindent
  
  \ifgenerate
    \ifcache
      \tikzsetnextfilename{irrev2to3_reM}
    \fi
    \input{\tikzpathirrev2to3_reM.tikz}
  \else
    \includegraphics{\tikzpathexternalirrev2to3_reM}
  \fi

\end{minipage}%
\begin{minipage}{\columnwidth}
  \centering
  
  \ifgenerate
    \ifcache
      \tikzsetnextfilename{irrev2to3_feMpartition}
    \fi
    \input{\tikzpathirrev2to3_feMpartition.tikz}
  \else
    \includegraphics{\tikzpathexternalirrev2to3_feMpartition}
  \fi

  \\\vspace{.3in}
  
  \ifgenerate
    \ifcache
      \tikzsetnextfilename{irrev2to3_reMpartition}
    \fi
    \input{\tikzpathirrev2to3_reMpartition.tikz}
  \else
    \includegraphics{\tikzpathexternalirrev2to3_reMpartition}
  \fi

\end{minipage}
\caption{The forward (top left) $M^+$ and reverse (bottom left) $M^-$ \eMs\ for
  a causally irreversible process. Note that $D^+ \neq D^-$ and, thus,
  $M^+ \neq M^-$.  The forward causal states $\ForwardCausalState$
  (top right) partition all allowable histories $\ldots X_{-1}$. In this
  example, the states are uniquely characterized
  by specifying the most recent symbol $X_{-1}$.  For example, any valid
  history ending with a $\Symbol{0}$ or $\Symbol{2}$ maps to state $A$, and
  the possible futures can begin with a $\Symbol{0}$ or $\Symbol{1}$. The
  incoming edges of $M^+$ correspond to histories ($X_{-1}$), while its outgoing
  edges correspond to futures ($X_{0}$).
  The reverse causal states $\ReverseCausalState$ (bottom right) partition
  all allowable futures $X_0...$. In this example, the states are uniquely
  characterized by specifying the earliest symbol of each future.  For example,
  any valid future beginning with a $\Symbol{0}$ maps to state $C$ and the
  associated histories must end with a $\Symbol{0}$ or $\Symbol{2}$. The incoming
  edges of $M^-$ correspond to futures ($X_0$), while its outgoing edges
  correspond to histories ($X_{-1})$.}
  \label{fig:IrreversibleExample}
\end{figure*}

Consider the time series over the alphabet $\{0,1,2\}$ whose forward ($M^+$)
and reverse ($M^-$) \eMs\ are shown in Fig.~\ref{fig:IrreversibleExample}.
As we will show, the process language generated by $M^+$ is irreversible and,
additionally, this irreversibility is due to an underlying topological
irreversibility. That is,
$\ForwardDFA \neq \ReverseDFA$ implies that $\ForwardEM \neq \ReverseEM$.

To see the topological irreversibility note that in $\FutureEM$ $w=01$ is a
valid word: Start in $A$, see $0$ and stay in $A$, then see $1$ and go to $B$.
However, $\widetilde{w} = 10$ is not a valid word. We can also see this in a
slightly different light by noting that $w$ is valid in $\FutureEM$, but not
valid in $\ReverseEM$.

To understand the forward causal states, consider the distribution
of $X_0=(0,1,2)$ conditioned on length-$1$ history suffixes:
\begin{align*}
  \Pr(X_0|X_{-1}=0) &= (1/2,1/2,0) ,\\
  \Pr(X_0|X_{-1}=1) &= (0, 1/2,1/2) , ~\text{and}\\
  \Pr(X_0|X_{-1}=2) &= (1/2,1/2,0) .
\end{align*}
We see that the time series generated by this machine has the following
characteristics: Every history that ends on symbol $0$ or $2$, is followed by
either $0$ or $1$, with probability $1/2$, but never by symbol $2$. Hence, with
regard to the distribution of a one-step future, all histories ending on $0$ or
$2$ are equivalent and we denote this class of equivalent
histories as causal state
$A$. The distribution of symbols following words ending on symbol $1$ is
different. They are followed by either symbols $1$ or $2$ with probability
$1/2$, but never by symbol $0$. All histories ending in $1$ are hence
equivalent with respect to the distribution of a one-step future and we
denote their equivalence class as state $B$.

States $A$ and $B$ partition of the entire space of allowable histories.
The fact that the equivalence class of a history is determined solely by the
last symbol is reflected by the time series of symbols having Markov order $1$.
The reader should verify that, in this particular example, Markovity also
means that the partition obtained by examining one-step futures is
equivalent to the partition obtained by examining arbitrary $L$-step futures.
From this, we see that $\ForwardCausalStateSet = X_{:0} / \sim^+$ consists of:
\begin{align*}
  A &= \{ \ldots 0, \:\ldots 2 \} ~\text{and}\\
  B &= \{ \ldots 1 \} ,
\end{align*}
where an ellipsis stands for any valid past.

The partition is represented graphically in the matrix at the top-right of
Fig.~\ref{fig:IrreversibleExample}.
In it, we independently rearranged the histories and futures so as
to cluster the block-structures within the matrix. For each history $X_{-1}$,
the distribution over futures $X_0$ is (topologically) represented as a column.
Histories with the same column colorings belong to the same equivalence class
under the forward equivalence relation $\sim^+$.
Finally, note that the futures are not partitioned by the forward equivalence
relation since $X_0=1$ is allowable from both $A$ and $B$.

To understand the reverse causal states, we examine the distribution of symbols
\emph{preceding} the future. Since the Markov order does not change when
analyzing the time series in the reverse direction (App. \ref{app:markov}), the
equivalence class of a future is determined solely by the first symbol of the
future.  Additionally, equality of distributions over length-$1$ histories
implies equality over arbitrary length-$L$ history distributions. Thus,
for $X_{-1} = (0,1,2)$ conditioned on a length-$1$ future, we have:
\begin{align*}
  \Pr(X_{-1}|X_{0}=0) &= (1/2,0,1/2) ,\\
  \Pr(X_{-1}|X_{0}=1) &= (1/4, 1/2,1/4) ,~\text{and}\\
  \Pr(X_{-1}|X_{0}=2) &= (0,1,0) .
\end{align*}

Any word starting with symbol $0$ can only be preceded by symbols $0$ or $2$
with probability $1/2$ each, but never with symbol $1$. Correspondingly, all
futures starting with symbol $0$ are equivalent and their equivalence class is
denoted as reverse causal state $C$. Furthermore, any word starting with symbol
$1$ is preceded by symbols $0$ or $2$ with probability $1/4$
each or is preceded by
symbol $1$ with probability $1/2$. All futures starting with symbol $1$ are
equivalent with respect the distribution of preceding symbols and subsumed as
reverse causal state $D$. Finally, words starting with symbol $2$ can only be
preceded by symbol $1$. The equivalence class of futures starting on symbol
$2$ is denoted reverse causal state $E$. From this, we see that
$\ReverseCausalStateSet = X_{0:} / \sim^-$ consists of:
\begin{align*}
  C &= \{ 0\ldots \} ,\\
  D &= \{ 1\ldots \} , ~\text{and}\\
  E &= \{ 2\ldots \} ,
\end{align*}
where an ellipsis now stands for any valid future.

States $C$, $D$, and $E$ partition the space of allowable futures. They are
represented in the lower-right matrix of Fig.~\ref{fig:IrreversibleExample}.
In it, we rearranged
the histories and futures so as to cluster the block-structures
within the matrix. For each future $X_0$, the distribution over histories
$X_{-1}$ is (topologically) represented as a row. Each row coloring is distinct,
reflecting the fact that each future belongs to a distinct reverse causal state
under the reverse equivalence relation $\sim^-$. Finally, note that the
histories are not partitioned by the reverse equivalence relation since
$X_{-1}=0$, for example, is allowable from both $C$ and $D$.

Note how the space of histories is partitioned into only two equivalence
classes, while the space of futures is partitioned into three equivalence
classes. Any first-order Markov chain on $k$ symbols has at most $k$ causal
states. That we only have two forward causal states is due to the
fact that the future distributions after seeing symbols $2$ and $0$ are
equivalent. This equivalence, however, does not hold in the reverse direction
and, so, there are three reverse causal states.
The asymmetry is further exemplified by the forward \eM\ having smaller
statistical complexity than the reverse \eM:
$C_\mu^+=1 ~\text{bit}~ < C_\mu^-=3/2 ~\text{bit}$.
For this particular process, it takes $1/2$ bit more memory, on average, to
generate the \emph{same} string of symbols from right to left than
from left to right.

Comparing the causal states as represented in
Fig.~\ref{fig:IrreversibleExample}, we see that each
equivalence relation also defines a partition over the set $(X_{-1},X_0)$.
This, in turn, extends to a partition over bi-infinite strings. So, we can
think of the forward (reverse) \eM\ as the restriction of this
partition to the set of histories (futures). The partition over the bi-infinite
strings must be such that when it is restricted to histories (futures) it
induces a unifilar dynamic over equivalence classes.  This particular point
will be important when we discuss the bidirectional machine in
Sec.~\ref{sec:bieM}.

\subsubsection{Probability-Driven Irreversibility}
\label{sec:ProbDrivenIrreversibility}

In our second example, we show that irreversibility can have purely
probabilistic origins.  We do this with an irreversible, order-$2$ Markov
process that has a reversible support. Figure~\ref{fig:ProbDrivenIrreversibility}
presents the recurrent components of the forward and reverse \eMs,
$\ForwardEM$ and $\ReverseEM$.  To see that the support is reversible, note
that the \eM\ structures, without probabilities, are equal:
$M^+_\emptyset = M^-_\emptyset$. This implies that $\ForwardDFA = \ReverseDFA$,
but it can also be seen directly since the topologies, in this example, are
already minimal deterministic finite automata. The practical consequence of
having a reversible support is that $\Pr(w) > 0$ if and only if
$\Pr(\widetilde{w}) > 0$.

\begin{figure}
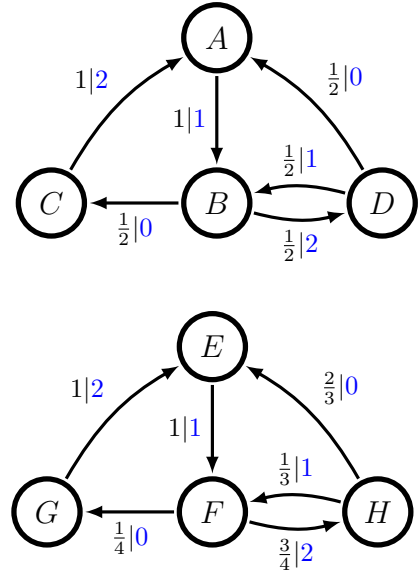

\noindent%
\centering

  \ifgenerate
    \ifcache
      \tikzsetnextfilename{probdriven_feM}
    \fi
    \input{\tikzpathprobdriven_feM.tikz}
  \else
    \includegraphics{\tikzpathexternalprobdriven_feM}
  \fi

\\\vspace{.2in}\noindent

  \ifgenerate
    \ifcache
      \tikzsetnextfilename{probdriven_reM}
    \fi
    \input{\tikzpathprobdriven_reM.tikz}
  \else
    \includegraphics{\tikzpathexternalprobdriven_reM}
  \fi

  \caption{The forward (top) $M^+$ and reverse (bottom) $M^-$ \eMs\ of an
    irreversible, order-$2$ Markov process. The process is irreversible since
    $\ForwardEM \neq \ReverseEM$. However, the support is reversible since the
    underlying topologies of each \eM\ are the same: $M^+_\emptyset =
    M^-_\emptyset$. }
\label{fig:ProbDrivenIrreversibility}
\end{figure}

Beginning with the forward causal states, we examine the distribution of symbols
that succeed histories. Since the process is order-$2$ Markovian, we calculate
finite histories and futures instead of semi-infinite histories and futures.
Specifically, partitioning length-$2$ histories based on the conditional
distributions of length-$2$ futures yields the same result as partitioning
semi-infinite histories based on the conditional distributions of arbitrary
length futures~\footnote{In this example, it is sufficient to consider
length-$1$ futures, but we use length-$2$ futures in order to demonstrate
the general technique.  That is, the columns of the matrix representing the
conditional distribution must be marginalized in order to obtain the transition
probabilities of the \protect\eM.}. We directly calculate
$\Pr( X_0,X_1 | X_{-2},X_{-1})$ as a right-stochastic matrix, finding:
\begin{align*}
  &\Pr( X_0,X_1 | X_{-2},X_{-1} ) \\
  &\quad=\bordermatrix{%
      & 01 & 02 & 10 & 12 & 20 & 21\cr
   01 & 0 & 1/2 & 0 & 0 & 1/4 & 1/4\cr
   02 & 0 & 0 & 1/2 & 1/2 & 0 & 0\cr
   10 & 0 & 0 & 0   & 0   & 0 & 1\cr
   12 & 1/2 & 0 & 1/4 & 1/4 & 0 & 0\cr
   20 & 0 & 0 & 1/2 & 1/2 & 0 & 0\cr
   21 & 0 & 1/2 & 0 & 0 & 1/4 & 1/4
  } .
\end{align*}
The forward causal states are groupings of histories and, in this presentation,
they correspond to groupings of identical rows. For example, the rows
corresponding to $01$ and $21$ are identical and, so, are grouped into the same
equivalence class. Translating these history suffixes back into semi-infinite
histories, we find that $\ForwardCausalStateSet = X_{:0} / \sim^+$ consists of:
\begin{align*}
  A &= \{ \ldots 02, \:\ldots 20 \} , \\
  B &= \{ \ldots 01, \:\ldots 21 \} , \\
  C &= \{ \ldots 10 \} , ~\text{and} \\
  D &= \{ \ldots 12 \} .
\end{align*}

The reverse causal states are similarly obtained, but now we consider the
distribution of symbols that precede futures. Once again, we work with
finite-length histories and futures. Using a right-stochastic matrix,
we calculate $\Pr(X_{-2},X_{-1} | X_0,X_1 )$ directly as:
\begin{align*}
  &\Pr( X_{-2},X_{-1} | X_{0},X_{1} ) \\
  &\quad=\bordermatrix{%
      & 01 & 02 & 10 & 12 & 20 & 21\cr
   01 & 0 & 0 & 0 & 1 & 0 & 0\cr
   02 & 1/4 & 0 & 0 & 0 & 0 & 3/4\cr
   10 & 0 & 1/2 & 0   & 1/4   & 1/4 & 0\cr
   12 & 0 & 1/2 & 0   & 1/4   & 1/4 & 0\cr
   20 & 1/4 & 0 & 0 & 0 & 0 & 3/4\cr
   21 & 1/12 & 0 & 2/3 & 0 & 0 & 1/4
  } .
\end{align*}
The reverse causal states are groupings of futures, and this corresponds to
groupings of identical rows in the matrix. Translating these
future prefixes into semi-infinite futures, we find that the reverse causal
states $\ReverseCausalStateSet = X_{0:} / \sim^-$ consist of:
\begin{align*}
  E &= \{ 02\ldots, \: 20\ldots \} , \\
  F &= \{ 10\ldots, \: 12\ldots \} , \\
  G &= \{ 01\ldots \} , ~\text{and} \\
  H &= \{ 21\ldots \} .
\end{align*}

Since there are multiple perspectives involved, we detour briefly to
translate the matrix $\Pr(X_{-2},X_{-1} | X_0,X_1)$ onto the reverse \eM\
shown in Fig.~\ref{fig:ProbDrivenIrreversibility}. One perspective, the global
perspective, is the process lattice of Fig.~\ref{fig:ProcessLattice}
that defines \emph{forward} as a left-to-right movement and \emph{reverse}
as a right-to-left movement. The other perspective, the local perspective, is
from the \eM's vantage point that is concerned only with its own \emph{local}
time. That is, the causal-state dynamic always proceeds ``forward'' in time,
irrespective of how forward is defined in the global perspective. For the
reverse \eM, this means its outgoing transitions translate to right-to-left
movements on the lattice. To demonstrate, consider the element:
\begin{align*}
  \Pr(X_{-2} = 2, X_{-1} = 1 | X_0 = 2,X_1 = 0) = 3/4 .
\end{align*}
The joint word is $x_{-2}x_{-1}x_0x_1 = 2120$.  To verify that this is a valid
word in the process, one scans the word from right-to-left following
transitions on $M^-$.  Focusing only on $x_0x_1=20$, if we begin in reverse
causal state $F$, then we transition to state $G$ on symbol $0$ and, finally, to
state $E$ on symbol $2$. This is precisely the statement of the reverse
causal-state partition: any future beginning with $20$ leads (when scanned from
right-to-left) to reverse causal state $E$. Continuing from $E$, we see
$x_{-2}x_{-1} = 21$ first by transitioning to state $F$ on symbol $1$ and then
again to state $H$ on symbol $2$. The total probability of this conditional path
is $3/4$.

To understand where the irreversibility arises, we first note that the matrix
$\Pr(X_0,X_1|X_{-2},X_{-1})$ implicitly contains the information about
the dynamic over the forward causal states. For example, from
$x_{-2}x_{-1} = 01$, we can see $x_0x_1=20$ and $x_0x_1=21$ each with probability
$1/4$.  Marginalizing and using the forward causal-state partition, this means
that state $B = \epsilon^+(\ldots 01)$ can see symbol $2$ with probability $1/2$
and when it does, we transition to state $D = \epsilon^+(\ldots 012) = \epsilon^+(\ldots 12)$.

Our goal is to understand why the edge from $F$ to $H$ on symbol $2$ occurs with
probability $3/4$ instead of probability $1/2$~\footnote{It is much easier to
see that the forward and reverse \protect\eMs\ are irreversible if matrix
$A \equiv \Pr(X_0,X_1|X_{-2},X_{-1})$ is compared to matrix
$B \equiv \Pr(\protect\widetilde{X}_0,\protect\widetilde{X}_1 |
              \protect\widetilde{X}_{-2},\protect\widetilde{X}_{-1})$,
instead of to matrix $C \equiv \Pr(X_{-2},X_{-1}|X_{0},X_{1})$. Matrices $A$
and $B$ are in the local time perspective and, thus, their forms are directly
comparable. Matrix $C$, in contrast, is in the global (lattice)
perspective of Fig.~\ref{fig:ProcessLattice} and requires index manipulation
to see that the resultant dynamics are irreversible.}.
From the reverse causal-state partition,
any future beginning with $X_0X_1=10$ will lead into state
$F = \epsilon^-(10\ldots)$.  If we then see $x_{-1} = 2$, we move to state
$H = \epsilon^-(210\ldots)$. In the matrix for
$\Pr(X_{-2},X_{-1} | X_0,X_1 )$, we now look at the row labeled $10$. There, the
columns labeled $02$ and $12$ correspond to histories with $x_{-1} = 2$.  The
probabilities are $1/2$ and $1/4$, respectively, which sum to $3/4$. So,
indeed, the process is irreversible, despite having a reversible support.

\subsubsection{Explosive Irreversibility}
\label{sec:FiniteToInfinite}

Our final example shows that, although a process can be represented by a
finite number of causal states in one direction, its presentation in the
reverse direction may require a countably infinite number of states. The
support of this process language corresponds to a strictly sofic
shift~\cite{Cove75} and, thus, the process is not Markovian. The consequence
is that we must use a hidden Markov model representation if we want to
represent it finitely, at least in the forward
direction~\footnote{Note that since the forward \protect\eM\ is finite, the
process \emph{does} have a finite reverse generator---namely, the time-reversed
forward \protect\eM. However, the minimality of the \protect\eM, within the
class of unifilar HMMs, ensures that this presentation can be smaller than
the reverse \protect\eM\ only if it is also nonunifilar.}. The recurrent
components of the forward and reverse \eMs\ are shown in
Fig.~\ref{fig:infiniteeM}.

\begin{figure*}
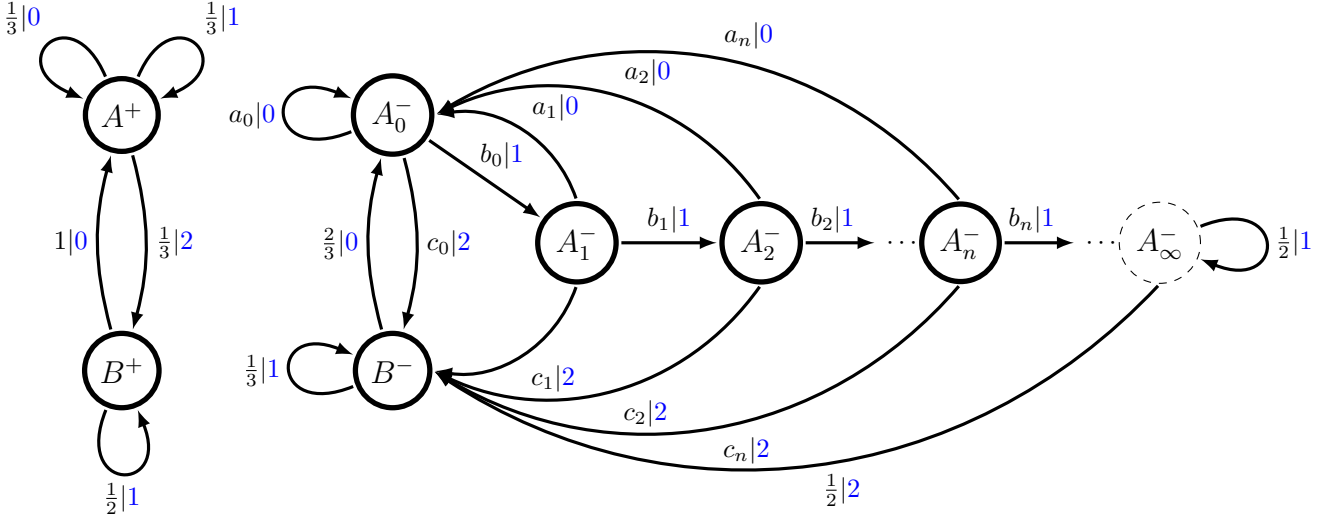

\centering

  \ifgenerate
    \ifcache
      \tikzsetnextfilename{explosive}
    \fi
    \input{\tikzpathexplosive.tikz}
  \else
    \includegraphics{\tikzpathexternalexplosive}
  \fi

\caption{Explosive irreversibility: Despite the forward \eM\ $M^+$ (left)
  having just two recurrent causal states, the reverse \eM\ $M^-$ (right)
  has a countable infinity of recurrent causal states. Transitions
  for $M^-$ make use of:
$a_n = 2^{n+1}(3z_n)^{-1}$,
$b_n = 1 - (a_n + c_n)$,
$c_n = 3^n (2z_n)^{-1}$, and
$z_n = 2^n + 3^n$.
  The dashed state labeled $A^-_\infty$ is an \emph{elusive} causal
  state~\cite{Uppe97a}; it is infinitely preceded, but neither reachable
  (from the omitted start state) nor recurrent. See App. \ref{app:infiniteeM}.
  }
\label{fig:infiniteeM}
\end{figure*}

Let us again study the distribution of symbols \emph{succeeding} histories.
Since the process is not Markovian, we cannot expect to obtain the causal states
by examining finite-length histories. And so, we must focus attention on
semi-infinite histories and their suffixes.

The presence of synchronizing words~\cite{Jame10a} makes the analysis a bit
easier. In this example, $w=0$ and $w=2$ are minimal synchronizing words and,
so, after observing one of these words, the state of the \eM\ is known with
certainty. \EM\ unifilarity then guarantees that on each next symbol we will
\emph{still} know the state of the machine.  This allows us to read the
distribution of $X_0$ directly off the forward \eM's outgoing edges.

Thus, any history ending in symbol $0$ will be followed by symbols $0$,
$1$, or $2$ with probability $1/3$ each.  The equivalence class of
histories containing $X_{:0} = \ldots 0$ will be denoted forward
causal state $A^+$.  Looking at
the machine, we see that the distribution of next symbols remains unchanged
whenever we see a $1$ from state $A^+$. So, any history ending in a $0$ followed
by an arbitrary, but finite number of $1$s also belongs to equivalence class $A^+$.
Similarly, any history ending with $2$ will be followed by symbols $1$ or $0$
with probability $1/2$ each. The equivalence class of histories ending in $2$
will be denoted forward causal state $B^+$ and, from the machine, we can also see
that $B^+$ includes any history ending with a $2$ followed by an
arbitrary, but
finite number of $1$s. A history consisting entirely of the symbol $1$ is best
understood by taking the limit of finite histories which also consist entirely
of $1$s. When one does this, the history will be followed by symbols $0$ or
$1$ with probability $1/2$ each. Concretely, for $X_0 = (0,1,2)$ and
$k \geq 0$, the conditional distributions for every valid history are:
\begin{align*}
  \Pr( X_0 | X_{:0} = \ldots 01^k) &= (1/3, 1/3, 1/3) ,\\
  \Pr( X_0 | X_{:0} = \ldots 21^k) &= (1/2, 1/2, 0)  , ~\text{and}\\
  \Pr( X_0 | X_{:0} = 1^\infty) &= (1/2, 1/2, 0) .
\end{align*}
And, from this, we see that $\ForwardCausalStateSet = X_{:0} / \sim^+$ consists of:
\begin{align*}
A^+ &= \{ \ldots 01^k \} ~\text{and}\\
B^+ &= \{ \ldots 21^k, 1^\infty \} .
\end{align*}

The distribution of symbols \emph{preceding} futures is more complicated.
First, we consider futures beginning with $1^k2$, $k \geq 0$. These futures
cannot be preceded by symbol $2$. The probability of observing a $0$ or
another $1$ preceding these futures is $2/3$ and $1/3$, respectively. We
denote the equivalence class of all futures starting with $1^k2$ as reverse
causal state $B^-$. Now, consider all words starting with $1^k0$, an arbitrary
number of $1$s followed by $0$. A short calculation shows that such words can
be preceded by a $0$, $1$, or $2$ with the probability depending
\emph{explicitly} on the number of $1$s at the beginning of the future. Thus,
there is one reverse causal state for every $k$, and we denote these states
as $A^-_k$. As before, the future consisting entirely of $1$s is most easily
understood by taking limits; one finds that it is not possible to precede the
future with a $0$ and that $1$ and $2$ precede the future with
probability $1/2$ each.
This limiting distribution coincides with $\lim_{k\to\infty} A^-_k$ and, so, we label its
equivalence class $A^-_\infty$. Formally, for $X_{-1} = (0,1,2)$ and $k \geq 0$,
the conditional distributions for every valid future are:
\begin{align*}
  \Pr( X_{-1} | X_{0:} = 1^k 2 \ldots) &= (2/3, 1/3, 0) ,\\
  \Pr( X_{-1} | X_{0:} = 1^k 0 \ldots) &= \textstyle
    \frac{\left(2^{k+2}, \, 2^{k+1} + 3^{k+1}, \,
	3^{k+1}\right)}{6(2^k + 3^k)} , ~\text{and}\\
  \Pr( X_{-1} | X_{0:} = 1^\infty ) &= (0, 1/2, 1/2) .
  \end{align*}
From this, we see that
$\ReverseCausalStateSet = X_{0:} / \sim^-$ consists of:
\begin{align*}
  A^-_0 &= \{0\ldots\} ,\\
  A^-_1 &= \{10\ldots\} ,\\
  A^-_2 &= \{110\ldots\} ,\\
  &\hspace{.08in}\vdots\\
  A^-_k &= \{1^k0\ldots\} ,\\
  &\hspace{.08in}\vdots\\
  A^-_\infty &= \{1^\infty\}\\
  B^- &= \{1^k2\ldots\} .
\end{align*}

Again, we leave it to the reader to verify that, for this particular example,
a partition of futures into equivalence classes with respect to the
preceding symbol will not change when considering longer strings of preceding
symbols.

The reverse causal states can also be obtained by applying the
forward causal-state equivalence relation on the time-reversed HMM of the
forward \eM.  That is, $(X_{0:} / \sim^-) \cong (\widetilde{X}_{:0} / \sim^+)$.
For example, reverse causal state $A^-_1$ contains every future beginning with
$10$.  Alternatively, we can associate $A^-_1$ with ``histories'' ($\widetilde{X}_{:0}$)
that end $01$.  Since the support is reversible, this allows for a direct
comparison to the forward causal-state partition, and so $A^-_1$ is a subset of
forward causal state $A^+$.  We summarize the relationship~\footnote{This
relationship is a comparison between the forward and reverse causal states only.
To each forward causal state, there is a 1-1 correspondence between its
histories and the union of futures from reverse causal states. Note that
this relationship says little about how the partitions are correlated in time.
For that, one must consider $\Pr(\ForwardCausalState | \ReverseCausalState)$.
See App.~\ref{app:infiniteeM}.} between the
partitions as follows:
\begin{align*}
  A^+ &\cong A^-_0 \cup A^-_1 \cup \cdots \cup A^-_k \cup \cdots ,\\
  B^+ &\cong B^- \cup A^-_\infty~.
\end{align*}

Recall, $M^+_\emptyset$ and $M^-_\emptyset$ denote the forward and reverse
DFAs whose structure is defined by the forward and reverse \eMs\ without
probabilities.  In this example, $M^+_\emptyset \neq M^-_\emptyset$ since they
disagree on the number of states. However, $M^-_\emptyset$ is not minimal
and would be equal to $M^+_\emptyset$, if it were minimized. This means that
the support of the process \emph{is} reversible: $\ForwardDFA = \ReverseDFA$.
Thus, this example also demonstrates probability-driven irreversibility, but
differs from the example in Sec.\ref{sec:ProbDrivenIrreversibility}, which
had $M^+_\emptyset = M^-_\emptyset$.

This example demonstrated that the \eMs\ of irreversible processes can be
finite in one direction and infinite in the other. The process has
$\ForwardCmu \approx 0.971$ and $\ReverseCmu \approx 1.589$ and, so once again,
we see that it takes more memory to generate the process from right-to-left
than from left-to-right.

\subsection{Survey of Irreversibility}

Reference~\cite{Crut10a} classified the space of hidden Markov models in terms
of unifilarity, synchronization, and minimality. Figure~\ref{fig:hmmspace}
reproduces the essential components of the hierarchy presented there, extending
it several ways \footnote{In Fig.~\ref{fig:hmmspace}, we stress that some atoms
may have zero measure. For example, every \protect\eM\ with uniformly distributed
transition probabilities is exactly synchronizing. Thus, the atom representing
hidden Markov models with uniformly distributed transition probabilities that
are simultaneously minimal unifilar and not exactly synchronizing is empty.}.

\begin{figure}
\centering

  \ifgenerate
    \ifcache
      \tikzsetnextfilename{hmmspace}
    \fi
    \input{\tikzpathhmmspace.tikz}
  \else
    \includegraphics{\tikzpathexternalhmmspace}
  \fi

\caption{Structural classification of hidden Markov models: Presentations within
  the green ellipse correspond to the recurrent \eMs.  The shaded area is
  the subset of recurrent \eMs\ that are exactly synchronizing and,
  additionally, have uniformly distributed transitions probabilities on the
  outgoing edges of each state. This subset defines the \emph{topological
  \eMs}. Areas in the diagram are not drawn to scale and only show which
  classes are contained in other classes.}
\label{fig:hmmspace}
\end{figure}

At the outer-most level, outside the dashed ellipse in Fig.~\ref{fig:hmmspace},
we have hidden Markov models that are strictly nonunifilar. So, given the
current state and symbol, there is residual uncertainty in the next state:
$H[\AlternateState_1 | \AlternateState_0, \MeasSymbol_0] > 0$. Moving inside
the dashed ellipse we encounter the strictly unifilar hidden Markov models for
which this quantity is exactly zero. Unifilarity is an important property since,
among other reasons, it allows one to calculate the process's
entropy rate directly from the presentation.

However, unifilar hidden Markov models can have a type of redundancy such that the state is not justified by the process statistics. Such models have gauge information
$\GI = H[\AlternateState_0 | X_{:0}, X_{0:}] > 0$. And, when we restrict
to those with $\GI = 0$, the hidden Markov models become
asymptotically synchronizing~\footnote{Reference~\cite{Crut10a} called this
class \vocab{weakly asymptotically synchronizing}, but it turns out to be
equivalent to (strongly) \emph{asymptotically synchronizing} \cite{Trav11a}.}.
This class exists within the dotted ellipse of Fig.~\ref{fig:hmmspace}. One
signature of unifilar models with zero gauge information is that the state
uncertainty vanishes asymptotically for almost every history in the process
language~\cite{Trav10b}.

Within the class of asymptotically synchronizing hidden Markov models, there
exists a subset for which the state uncertainty vanishes in finite time for
almost every history in the process language~\cite{Trav10a}. Such hidden Markov
models necessarily have at least one synchronizing word. In
Fig.~\ref{fig:hmmspace}, this is delineated by the blue ellipse.

Another subset within the class of asymptotically synchronizing hidden Markov
models are the minimal unifilar hidden Markov models.  Any hidden Markov
model with these properties corresponds to an \eM\ of a process
language~\cite{Trav10d}. This is represented by the
green ellipse in Fig.~\ref{fig:hmmspace}. Generally, the set of \eMs\ and
the set of exactly synchronizing hidden Markov models (blue
ellipse) are not the same, and
their intersection defines the class of exactly synchronizing \eMs.

Reference~\cite{Crut10a}'s classification of processes and their presentations
provides a natural setting for developing a refined classification based on the
irreversibility properties just introduced. As a first step, though, it is
perhaps more helpful to develop a quantitative appreciation of how common
irreversibility is within the space of hidden Markov models. This is a
difficult, if somewhat open-ended challenge, but we can make some progress by
examining several subclasses. Systematically surveying presentations is
generally difficult due to the probabilistic nature of hidden Markov models and
the processes they generate. However, if we restrict ourselves to hidden Markov
models with uniformly distributed transition probabilities leaving each
state---recall the red, wavy parabola in Fig.~\ref{fig:hmmspace}---then we can
systematically enumerate them. Essentially, the task boils down to enumerating
a particular class of finite-state automata. Reference~\cite{Johnson2010a}
provided an exhaustive enumeration of exactly synchronizing \eMs\ with
uniformly distributed transition probabilities leaving each state. It is this
class of processes---generated by the topological \eMs---that we survey in order to develop
an appreciation of how common irreversibility is within the space of hidden Markov models.

Table \ref{tab:topoeMs} summarizes the survey, giving the number
$N_{n,k}$ of topological \eMs\ \cite{Johnson2010a} and the number $C_{n,k}$
of irreversible \eMs\ over $n$ states and exactly $k$ symbols in the alphabet.
(By ``exactly $k$ symbols'' we emphasize that we excluded from the counts
processes with $k = 3$ that use only $2$ symbols, for example.)
The immediate impression is quite striking: Irreversibility dominates. It
comprises over $98\%$ of all topological \eMs\ and their
associated processes. Indeed, the fraction of irreversible
\eMs\ appears to rapidly increase toward unity as the number of states increases.
And so, what might have initially appeared to be a counterintuitive
property---temporal asymmetry in the statistics of a stationary process---is
the overwhelming rule in the space of processes.

\begin{table*}
  \centering
  \begin{tabular}{|c|rr|rr|rr|rr|rr|}
  \hline
  $n\backslash k$ & $N_{n,2}$ & $C_{n,2}$ & $N_{n,3}$ & $C_{n,3}$
  & $N_{n,4}$ & $C_{n,4}$ & $N_{n,5}$ & $C_{n,5}$ & $N_{n,6}$ & $C_{n,6}$ \\ \hline
  1 & 1 & 0 & 1 & 0 & 1 & 0 & 1 & 0 & 1 & 0 \\
  2 & 7 & 0 & 120 & 84 & 1,351 & 1,200 & 12,900 & 12,290 & 113,827 & 111,390 \\
  3 & 78 & 24 & 15,364 & 14,561 & 1,596,682 & 1,586,736 \\
  4 & 1,388 & 1,077 & 3,621,474 & 3,607,084 \\
  5 & 35,186 & 33,107\\
  6 & 1,132,613 & 1,119,623
  \end{tabular}
\caption{The number $N_{n,k}$ of topological \eMs\ \cite{Johnson2010a} and the
  number $C_{n,k}$ of irreversible \eMs\ over $n$ states and exactly
  $k$ symbols in the alphabet.
  }
\label{tab:topoeMs}
\end{table*}

\section{The Bidirectional Machine}
\label{sec:bieM}

The process, as a stationary probability space, is a bulky abstraction, and
state-based models, such as hidden Markov models, are often used to provide
a much more concise representation.  However, the forward and reverse
generators of a process are not unique, and this makes it difficult to separate
structure in the process from structure in presentations of the process.
The entropy rate $\hmu$ and excess entropy $\EE$ are two well known structural
properties of a process. We showed, in addition, that crypticity $\PC$,
oracular information $\OI$, and gauge information $\GI$ are important
structural properties of presentations.

The forward and reverse \eMs\ were introduced as a process's canonical
presentations and, in doing so, the statistical complexities $\Cmu^+$ and
$\Cmu^-$ became process properties that, in addition, were easily accessible
through these privileged presentations. The \eMs\ were ideal in a number of
ways, for example and importantly, they provided a direct calculation of a
process's entropy rate. The excess entropy, however, remained inaccessible
and, so, a new presentation was required.

The bidirectional machine, introduced in Refs.~\cite{Crut08a, Crut08b}, is
a generator that unites the forward and reverse \eMs, providing an
explicit accounting of the relationship between them~\footnote{Given the forward
\protect\eM\ of a process, one can construct its reverse \protect\eM\ using
the technique described in Ref.~\cite{Crut08b}. From this construction, we
learn how the forward and reverse causal states are related. However, if one is
given \emph{only} the reverse \protect\eM, then this important information
is lost and must be deduced again. The bidirectional machine is a presentation
of the process that preserves this information.}. In doing so,
the excess entropy, a structural property of a process, becomes accessible
through a simple calculation and, further, the bidirectional machine
\emph{contains} all information necessary to reconstruct the forward and
reverse \eMs. In this section, we define the bidirectional machine and
interpret it through an example from the previous section.

\subsection{Definitions}
The hidden process lattice of Fig.~\ref{fig:ProcessLattice} invites us to
consider a dynamic over joint causal states. We define an aggregate
state $\CausalState^\pm \equiv (\CSjoint)$ as the $2$-tuple of the
forward and reverse causal states with stationary distribution function:
\begin{align*}
 \pi(\alpha\gamma) & \equiv \Pr\left(\CausalState^\pm = (\alpha, \gamma) \right)\\
	& = \pi(\alpha,\gamma) \\
    & \equiv \Pr(\CausalState^+ = \alpha, \CausalState^- = \gamma),
\end{align*}
for $\alpha \in \CausalState^+$ and $\gamma \in \CausalState^-$. Counter
to typical usage $\pm$ in the joint causal state is interpreted as forward
\emph{and} reverse, rather than \emph{or}. Note, that we purposefully overload
notation and use $\pi$ again, but it will always be clear from context to
which generator we refer.

Given the (stationary) distribution $\pi$, if we scan left-to-right, we obtain
a forward generator $M^\pm$ of the process. If we scan right-to-left, we obtain
the process's reverse generator $M^\mp$. These generators are generally
distinct.  However, we will see that $M^\mp$ is equal to the time-reversed
HMM of $M^\pm$.  That is, $M^\mp = \widetilde{M}^\pm$. For that reason, we take
$M^\pm$ as the starting point.

Having defined the states, the transition matrices for the
\vocab{forward bidirectional machine} $M^\pm$ are given by:
\begin{align}
\label{eq:bieM}
  T_x(\alpha\gamma, \beta\delta)
  &\equiv \Pr\left( X_0=x, \CausalState_1^\pm = (\beta,\delta) \,|\,
                    \CausalState_0^\pm = (\alpha,\gamma) \right) \nonumber\\
  &= \begin{cases}
       \widetilde{T}_x(\gamma, \delta) & \text{if } T_x(\alpha, \beta) > 0,\\
       0 & \text{otherwise},
     \end{cases}
\end{align}
where $\alpha,\beta \in \CausalStateSet^+$ and
$\gamma,\delta \in \CausalStateSet^-$.  The transition probabilities of the
forward bidirectional machine mimic the transition probabilities of the
time-reversed reverse \eM\ ($\widetilde{M}^-$), provided the transition
is allowed in the forward \eM\ ($M^+$).

To see how Eq.~\eqref{eq:bieM} arises, first we note that:
\begin{align}
\label{eq:bieMfactoring}
& \hspace{-.1in}
  \Pr\left( X_0, \CausalState_1^\pm  \,|\,
                   \CausalState_0^\pm \right) \nonumber \\
&= \Pr\left( X_0, \CausalState_1^+, \CausalState_1^- \,|\,
                   \CausalState_0^+, \CausalState_0^- \right) \nonumber \\
&= \Pr\left( \CausalState_1^+ \,|\,
             \CausalState_0^+, \CausalState_0^-, X_0, \CausalState_1^-\right)
   \Pr(X_0, \CausalState_1^- \,|\, \CausalState_0^+, \CausalState_0^-) ~.
\end{align}
Following Eq.~\eqref{eq:bieM}, we take $\CausalState_0^+=\alpha$,
$\CausalState_1^+=\beta$, $\CausalState_0^-=\gamma$, $\CausalState_1^-=\delta$,
and $X_0 = x$.  Then, the first factor in Eq.~\eqref{eq:bieMfactoring} is
either $0$ or $1$, due to unifilarity of the forward \eM, depending on
if $\beta$ is the unique causal state that follows $\alpha$ on symbol $x$.
The presence of $\CausalState_0^- = \gamma$ and $\CausalState_1^- = \delta$
in the conditional does not change this fact, so long as $(\gamma,x,\delta)$ is
a valid consecutive combination in the reverse \eM---and this is implicitly
handled by the second factor.

The second factor reduces due to the shielding property of hidden Markov models:
the \vocab{past} and \vocab{future} are independent given the \vocab{present}
state. Focusing on the reverse \eM, we express independence formally as:
\begin{align*}
  \Pr( X_{0:}^{\phantom{+}}, \ReverseCausalState_{1:} |
       X_{:0}^{\phantom{+}}, \ReverseCausalState_{:0}, \ReverseCausalState_0)
   = \Pr( X_{0:}^{\phantom{+}}, \ReverseCausalState_{1:} |
          \ReverseCausalState_0)
\end{align*}
where $(X_{0:}^{\phantom{+}}, \ReverseCausalState_{1:})$ is everything related
to the future and $(X_{:0}^{\phantom{+}}, \ReverseCausalState_{:0})$ is
everything related to the past.  Now, we also know that the forward causal
states are determined by the past:
$H[\ForwardCausalState_0 | X_{:0}^{\phantom{+}}] = 0$, and this means that the
forward causal state and future are independent given the past causal
state~\footnote{Let $A$, $B$, $C$, and $D$ be random variables
such that $A$ maps deterministically onto $D$.  Further, suppose that $A$ and
$B$ are independent given $C$. Then it follows that $D$ and $B$ are also
independent given $C$.}:
\begin{align*}
  \Pr( X_{0:}^{\phantom{+}}, \ReverseCausalState_{1:} |
       \ForwardCausalState_0, \ReverseCausalState_0)
   = \Pr( X_{0:}^{\phantom{+}}, \ReverseCausalState_{1:} |
          \ReverseCausalState_0)
\end{align*}
Restricting to single-step futures, we obtain:
\begin{align}
\label{eq:shielding}
  \Pr(X_0, \CausalState_1^- \,|\, \CausalState_0^+, \CausalState_0^-)
  =
  \Pr(X_0, \CausalState_1^- \,|\, \CausalState_0^-).
\end{align}
Understanding this in terms of previously defined quantities is
subtle precisely due to the shifting notions of forward and reverse
time. Intuitively, Fig.~\ref{fig:ProcessLattice} shows that we are asking the
reverse \eM\ to move left-to-right.  This direction is opposed to the
reverse \eM's local notion of forward time. Thus, we expect this movement from
left-to-right to relate to the time-reversed transition matrices of the
reverse \eM.

At a lower level, we note that the definition, Eq.~\eqref{eq:treM}, of the
time-reversed hidden Markov model was stated under the assumption that the
original model's \emph{increasing} indexes corresponded to a left-to-right
movement on the lattice. From the labeling in Fig.~\ref{fig:ProcessLattice},
the reverse \eM\ does not satisfy this assumption, and a proper translation of
Eq.~(\ref{eq:treM}) is:
\begin{align}
  \nonumber
  \widetilde{T}_x(\gamma,\delta)
  &\equiv \Pr(X_0=x, \CausalState^-_1=\delta\,|\, \CausalState^-_0=\gamma) \\
  &= \frac{ \pi(\delta) T_x(\delta, \gamma) }{ \pi(\gamma) },
\end{align}
which is exactly the quantity in question. The result is that whenever
$T_x(\alpha,\beta) > 0$, then the transition probability of the forward
(left-to-right) bidirectional machine is determined by the transition
matrices of the time-reversed reverse \eM\ ($\widetilde{M}^-$).

\begin{figure}
\centering
\includegraphics{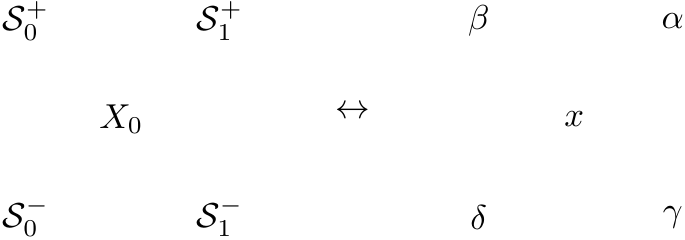}\\
\caption{
  Left: Portion of the process lattice relevant to the
  bidirectional machine's transition matrices. Right: Realizations
  of the process lattice as it applies to Eqs. \eqref{eq:trbieM1},
  \eqref{eq:trbieM2pre}, \eqref{eq:trbieM2}, \eqref{eq:samepath1}, and
  \eqref{eq:samepath2}.
}
\label{fig:Lattice2}
\end{figure}

The \emph{reverse bidirectional machine} $M^\mp$ is analogously defined by the
right-to-left dynamic over the joint causal states. This requires that the
forward \eM\ move right-to-left on the lattice, a direction that is opposed to
its local sense of forward time. The result is that we use the time-reversed
forward \eM\ ($\widetilde{M}^+$). Similarly, the reverse \eM\ is required to
move right-to-left on the lattice. This direction is in agreement with its
local sense of forward time and so, we utilize the reverse \eM\ ($M^-$) as is.
For $\alpha, \beta \in \CausalStateSet^+$ and $\gamma, \delta \in
\CausalStateSet^-$, we have:
\begin{align}
  \widetilde{T}_x(\alpha\gamma,\beta\delta)
  &\equiv
  \Pr\left( X_0=x, \CausalState_0^\pm = (\beta,\delta) \,|\,
                   \CausalState_1^\pm = (\alpha,\gamma) \right) \nonumber\\
  \label{eq:trbieM1}
  &= \begin{cases}
       \widetilde{T}_x(\alpha, \beta) & \text{if } T_x(\gamma, \delta) > 0,\\
       0 & \text{otherwise}.
     \end{cases}
\end{align}
The proof proceeds analogously to the forward bidirectional machine and is
omitted here.  For future reference, Fig.~\ref{fig:Lattice2} displays
$\alpha$, $\beta$, $\gamma$,
and $\delta$ on the process lattice, as they are used in the definition of
the reverse bidirectional machine. Thus, we have $\CausalState_0^\pm =
(\beta, \delta)$ and $\CausalState_1^\pm = (\alpha, \gamma)$. Note, these
variables are swapped in the definition of the forward
bidirectional machine.

The choice of $\widetilde{T}_x$ as the notation for the reverse bidirectional
machine's transition matrices suggests that it is related to the time-reversal
of the forward bidirectional machine.  Indeed, the definition of the
forward bidirectional machine already provides the matrices for the
right-to-left dynamic.  Thus, we see that $M^\mp = \widetilde{M}^\pm$:
\begin{align}
  \widetilde{T}_x(\alpha\gamma,\beta\delta)
  &\equiv
  \Pr\left( X_0=x, \CausalState_0^\pm = (\beta,\delta) \,|\,
                   \CausalState_1^\pm = (\alpha,\gamma) \right) \nonumber\\
  \label{eq:trbieM2pre}
  &= \frac{\pi(\beta, \delta) T_x(\beta\delta, \alpha\gamma)}{\pi(\alpha,\gamma)} ~.
\end{align}
Applying Eq.~\eqref{eq:bieM} gives the direct relation to the forward \eM:
\begin{align}
  \widetilde{T}_x(\alpha\gamma,\beta\delta)
  &= \begin{cases}
    \frac{\displaystyle\pi(\beta,\delta)\rule[-4pt]{0pt}{0pt}}
         {\displaystyle\pi(\alpha,\gamma)\rule[0pt]{0pt}{9pt}}
         \widetilde{T}_x(\delta, \gamma)
    &\text{if } T_x(\beta, \alpha) > 0, \\
    0 & \text{otherwise};
  \end{cases} \nonumber\\
  \label{eq:trbieM2}
  &= \begin{cases}
    \frac{\displaystyle\pi(\beta \,|\, \delta)\rule[-4pt]{0pt}{0pt}}
         {\displaystyle\pi(\alpha \,|\, \gamma)\rule[0pt]{0pt}{9pt}}
         T_x(\gamma, \delta)
    &\text{if } \widetilde{T}_x(\alpha, \beta) > 0, \\
    0 & \text{otherwise}.
  \end{cases}
\end{align}

Comparing Eqs.~\eqref{eq:trbieM1} and \eqref{eq:trbieM2}, we see that
whenever $T_x(\gamma, \delta)$ and $\widetilde{T}_x(\alpha,\beta)$ are
simultaneously positive, then we have:
\begin{align}
 \label{eq:samepath1}
 \pi(\alpha \,|\, \gamma) \widetilde{T}_x(\alpha, \beta) =
 \pi(\beta \,|\, \delta) T_x(\gamma, \delta) ~.
\end{align}
A complementary relation, obtained by applying Bayes theorem, is:
\begin{align}
 \label{eq:samepath2}
 \pi(\gamma \,|\, \alpha) {T}_x(\beta, \alpha) =
 \pi(\delta \,|\, \beta) \widetilde{T}_x(\delta, \gamma) ~.
\end{align}

The interpretations of Eqs.\eqref{eq:samepath1} and \eqref{eq:samepath2}
are properly framed using the process lattice, as shown in
Fig.~\ref{fig:Lattice2}. We could have
also worked with the forward bidirectional machine, expressing its transition
matrix as the Bayes inverse of $\widetilde{T}_x$ and, then, equating it to
Eq.~\eqref{eq:bieM}.  However, this does not yield any new insight.

Generally, these equations represent path equivalence on the process lattice.
In the left-hand side of Eq.~\eqref{eq:samepath1}, we begin in
$\CausalState_1^-=\gamma$, transition to $\CausalState_0^-=\delta$
on symbol $\MeasSymbol_0=x$, and then shift to $\CausalState_0^+=\beta$.
This path is represented in red in the right diagram of
Fig.~\ref{fig:pathequiv}.  The right-hand side of Eq.~\eqref{eq:samepath1}
says that the red path is equivalent (in probability) to the blue path, which
also begins in $\CausalState_1^-=\gamma$. However, now it shifts to
$\CausalState_1^+ = \alpha$ first, and then reverse transitions to
$\CausalState_0^+=\beta$ on symbol $\MeasSymbol_0=x$.
Equation~\eqref{eq:samepath2} provides an analogous result and is summarized
in the left diagram of Fig.~\ref{fig:pathequiv}.  There, we begin in
$\CausalState_0^+$ and transition to $\CausalState_1^-$ via two equivalent paths.

\begin{figure}
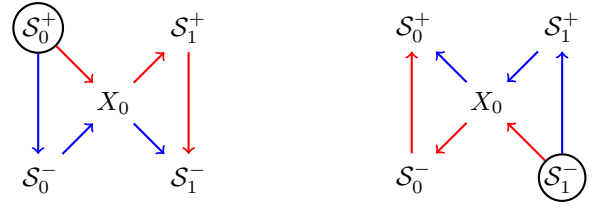

\centering

  \ifgenerate
    \ifcache
      \tikzsetnextfilename{pathequiv1}
    \fi
    \input{\tikzpathpathequiv1.tikz}
  \else
    \includegraphics{\tikzpathexternalpathequiv1}
  \fi

$\qquad\phantom{\leftrightarrow}\qquad$

  \ifgenerate
    \ifcache
      \tikzsetnextfilename{pathequiv2}
    \fi
    \input{\tikzpathpathequiv2.tikz}
  \else
    \includegraphics{\tikzpathexternalpathequiv2}
  \fi

\caption{Equations~\eqref{eq:samepath1} (left) and \eqref{eq:samepath2} (right)
demonstrate path equivalence. In each, the red and blue paths are
equivalent ways of moving around on the process lattice.}
\label{fig:pathequiv}
\end{figure}

The bidirectional machines are so-named because their state space consists
of the forward and reverse causal states and their transition dynamic allows
one to go in either direction. However, the bidirectional machine is still a
one-way generator and this is why there are two variants, $M^\pm$ and $M^\mp$.
These two variants are simply time-reversals of one
another, even if the underlying process is irreversible. Having established
the proper setting for bidirectional machines, in the next section we analyze
several of their properties and the consequences of this symmetry.

\subsection{Bi-Infinite Sequence Partitions}

The bidirectional machine $M^\pm$ can be understood by analyzing its effect
on the past and future.  Previously, we saw that the forward causal states
$\ForwardCausalState$ partitioned the semi-infinite histories $X_{:0}$, while
the reverse causal states $\ReverseCausalState$ partitioned the semi-infinite
futures $X_{0:}$. The bidirectional machine, it turns out, partitions the set
of bi-infinite strings $X_{:} \equiv X_{:0}X_{0:}$.  This is expressed by the
bidirectional equivalence relation~\cite{Crut08b}:
\begin{align*}
  (x_{:0}, x_{0:}) \sim^\pm (x_{:0}^\prime, x_{0:}^\prime) \iff
  \epsilon^+(x_{:0}) &= \epsilon^+(x_{:0}^\prime) \text{ and } \\
  \epsilon^-(x_{0:}) &= \epsilon^-(x_{0:}^\prime) ~.
\end{align*}
Thus, the bidirectional causal states $\BiCausalStateSet$ are a partition of
bi-infinite strings resulting from the application of an equivalence
relation: $\CausalStateSet^\pm = (X_{:0}, X_{0:}) / \sim^\pm$.  The mapping
$\epsilon^\pm(\cdot)$ that takes a bi-infinite string to its bidirectional
causal state is defined:
  \begin{align*}
  &\hspace{-.15in}\epsilon^\pm (x_{:0}, x_{0:}) \\
  &= \{ (x_{:0}^\prime, x_{0:}^\prime) :
    x_{:0}^\prime \in \epsilon^+(x_{:0}) \text{ and }
    x_{0:}^\prime \in \epsilon^-(x_{0:}) \} ~.
\end{align*}
Note, that the same equivalence relation is used for the forward
and reverse bidirectional machines.  All that differs is the dynamic
over the states.

For $M^\pm$ and $M^\mp$, we use a bi-infinite instance and shift
the time origin to the right (for $M^\pm$) or to the left (for $M^\mp$).
The symbol encountered during the shift is the symbol generated.
%\begin{align*}
%  M^\pm \quad&: \quad\epsilon^\pm(x_{:0}, x_{0:})
%    \xrightarrow {x_0} \epsilon^\pm(x_{:1}, x_{1:})\\
%  M^\mp \quad&: \quad\epsilon^\pm(x_{:1}, x_{1:})
%    \xrightarrow{x_0} \epsilon^\pm(x_{:0}, x_{0:})
%\end{align*}

However, given any bidirectional partition, it does not follow that the dynamic will be
unifilar, and this is precisely the case for the bidirectional machines.  With
\eMs, all histories (or futures) in the equivalence class have exactly the same
distribution over futures (or histories). And so, on the next symbol, every
history (or future) in the causal state transitioned to the same next causal
state.  With the bidirectional machine, this is no longer true, and the
dynamic over the states is generally nonunifilar.

\subsection{Properties}

Each of the process and presentation properties discussed can be considered
operators.  That is, given a model $M$, we calculate a quantity relative
to the model alone, using its \emph{local} sense of time.  This point is
worth remembering as we discuss properties of the bidirectional machines. We
will continue, however, to frame the various quantities using the
bird's eye view of the process lattice.

The stationary distribution for the forward bidirectional machine is
$\Pr(\ForwardCausalState,\ReverseCausalState)$ and, as
Sec.~\ref{sec:HMMReversibility} discussed, the reverse bidirectional machine
has the same stationary distribution. Using Refs.~\cite{Crut08a,Crut08b},
we can immediately calculate the excess entropy as $\EE = I[\CSjoint[;]]$.
Importantly, this quantity is not calculable given \emph{only} the forward and
reverse \eMs. (Alternate methods to calculate $\EE$ end up being essentially
equivalent to invoking the bidirectional machine.)

As mentioned, the bidirectional machine can also be nonunifilar. Since the
bidirectional causal states are the joint distribution over the forward and
reverse causal states, the bidirectional machine's oracular information
$\OI(M^\pm)$ is the crypticity of the reverse \eM\ $\PC^-$. Additionally, the
bidirectional machine's crypticity $\PC(M^\pm)$ is the crypticity $\PC^+$ of
the forward \eM.

If, instead, we work with the reverse bidirectional machine $M^\mp$, all the
interpretations are flipped. Then the crypticity $\PC(M^\mp)$ is equal to the
reverse \eM's crypticity $\PC^-$, and the oracular information $\OI(M^\mp)$
is the forward \eM's crypticity $\PC^+$. Recall that \eMs\ do not have oracular
information, since they are unifilar.

These information quantities are summarized in
Fig.~\ref{fig:BidirectionalIDiagram}. There, we see that the reverse
bidirectional machine swaps crypticity and oracular information just as a
general hidden Markov model~\cite{Crut10a}.

Of the presentation quantifiers, this leaves only the gauge information
$\GI^\mp$ to be explained. Recall that the past $X_{:0}$ completely determines
the future causal state $\ForwardCausalState_0$ and that the future $X_{0:}$
completely determines the past causal state $\ReverseCausalState$. Then, this
gives:
\begin{align*}
    \GI(M^\pm)
  &=
    H[\CausalState^\pm | X_{:0}, X_{0:}]\\
  &=
    H[\CSjoint[,] | X_{:0}, X_{0:}] \\
  &=
    H[\ForwardCausalState | X_{:0}, X_{0:}]
    + H[\ReverseCausalState | X_{:0}, X_{0:}, \ReverseCausalState] \\
  &\leq
    H[\ForwardCausalState | X_{:0}] + H[\ReverseCausalState | X_{0:}] \\
  &= 0 + 0 ~.
\label{eq:BiMachineNoGauge}
\end{align*}
Thus, the bidirectional machine does have a certain representational efficiency:
It has no gauge information. This is implicitly shown in
Fig.~\ref{fig:BidirectionalIDiagram}, but more easily seen in
Fig.~\ref{fig:eMidiagram}.  There, we see that the ellipse representing the
bidirectional machine's states (the union of $\Cmu^+$ and $\Cmu^-$) only
consists of areas within the entropies of the past $H[X_{:0}]$ and
future $H[X_{0:}]$. Naturally, one wonders if it is possible to define
the bidirectional machine through constraints.  To this end, we conjecture that the
bidirectional machine is the only generator of the process with zero gauge
information that marginalizes into the forward and reverse \eMs\ and, additionally,
has $\PC(M^\pm) = \PC(M^+)$ and $\OI(M^\pm) = \PC(M^-)$.

We now turn to the various state entropy quantities that play a role in the
bidirectional machine. The state entropy of the forward and reverse \eMs\
represented the forward and reverse statistical complexities:
$\ForwardCmu \equiv H[\ForwardCausalState]$ and
$\ReverseCmu \equiv H[\ReverseCausalState]$.  Similarly, we denote the
state entropy of the forward and reverse bidirectional machine by
$\Cmu^\pm = H[\CausalState^\pm] = H[\CSjoint[,]]$ and call it the
\vocab{bidirectional statistical complexity}.
It represents the total amount of information needed to
predict or retrodict optimally. The key difference
between $\Cmu^\pm$ and the directed statistical complexities is that with the
bidirectional machine, one has a choice in which action, prediction or
retrodiction, is taken~\footnote{One
must choose only one action: prediction or retrodiction.  The forward
bidirectional machine allows one to make a prediction, while the reverse
bidirectional machine allows one to make a retrodiction.  Making a simultaneous
prediction and retrodiction with each machine does not yield the correct
joint probabilities over predicted and retrodicted symbols}.   We further note
that both $\Cmu^+$ and $\Cmu^-$ play equivalent roles in $\Cmu^\pm$, to the
extent that $\EE$ is contained in both. Due to this, we can see that:
\begin{align}
  \Cmu^\pm = \Cmu^+ + \Cmu^- - \EE ~.
\end{align}

One can also marginalize the bidirectional machine's transition matrices to
recover the forward and reverse \eMs. For $\alpha,\beta
\in \CausalStateSet^+$ and $\delta, \gamma \in \CausalStateSet^-$, we
marginalize $M^\pm$ to get $M^+$ as follows:
\begin{align*}
  T_x(\alpha,\beta)
    &= \Pr( X_0 = x, \ForwardCausalState_1 = \beta
         \,|\, \ForwardCausalState_0 = \alpha) \\
    &= \sum_{\gamma,\delta}
	\pi(\gamma \,|\, \alpha) \, T_x(\alpha\gamma, \beta\delta),
\end{align*}
where $\pi(\gamma \,|\, \alpha) \equiv \pi(\alpha\gamma) / \pi(\alpha)$
and $T_x(\alpha\gamma, \beta\delta)$ is given by Eq.~\eqref{eq:bieM}.
Similarly, we marginalize $M^\mp$ to get $M^-$:
\begin{align*}
  T_x(\gamma,\delta)
    &= \Pr( X_0 = x, \ReverseCausalState_0 = \delta
            \,|\, \ReverseCausalState_1 = \gamma)\\
    &= \sum_{\alpha,\beta}
	\pi(\alpha \,|\, \gamma) \, \widetilde{T}_x(\alpha\gamma, \beta\delta),
\end{align*}
where $\pi(\alpha \,|\, \gamma) \equiv \pi(\alpha\gamma) / \pi(\gamma)$
and $\widetilde{T}_x(\alpha\gamma, \beta\delta)$ is given by
Eq.~\eqref{eq:trbieM1}.

It also happens that knowing the bidirectional causal state is not always
helpful. Specifically, we have:
\begin{align*}
H[X_0|\ForwardCausalState_0,\ReverseCausalState_0]
  &= H[X_0 |
\ReverseCausalState_0] ~\mathrm{and}\\
H[X_{-1}|\ForwardCausalState_0,\ReverseCausalState_0]
  &= H[X_{-1} | \ForwardCausalState_0] ~.
\end{align*}
In other words, a question about the future is best understood by something
which comes from the future (and vice versa for questions about the past).
The reason for each of these results can be immediately deduced from
Fig.~\ref{fig:eMidiagram}.

\subsection{Uses}

The bidirectional machine is also useful in a number of ways. We
briefly mention several.

First, we note that $M^\pm$ and $M^\mp$, together, could be interpreted as a
transducer. Given a desired direction of time, one can move forward or backward
along the process lattice.  While the transducer viewpoint holds for any hidden
Markov model, only the bidirectional machine allows one to predict or
retrodict. To wit, if one constructed a transducer using $M^+$ and
$\widetilde{M}^+$, then one could make predictions, but it would not be
possible to retrodict since the forward causal states are not sufficient
statistics for the future---they are not suited for retrodiction.  This is
precisely the advantage of the bidirectional machine, since it tracks both
the forward and reverse causal states.

Second, the bidirectional machine allows one to exactly calculate the
persistent mutual information $\mathcal{I}_1$ \cite{Ball2010} over
a single-step time interval. Previously available only through empirical
estimates, $\mathcal{I}_1$ is the amount of information
$I[\MeasSymbol_{:0};\MeasSymbol_{1:}|\MeasSymbol_0]$
shared between $\MeasSymbol_{:0}$ and $\MeasSymbol_{1:}$, ignoring
$\MeasSymbol_0$. Note that neither \eM\ can give us the appropriate
distribution over $\MeasSymbol_{:0}$ and $\MeasSymbol_{1:}$, but the
bidirectional machine can. And so, it allows one to calculate $\mathcal{I}_1$
exactly. Since $\MeasSymbol_{:0}$ determines $\CausalState_0^+$ and
$\MeasSymbol_{1:}$ determines $\CausalState_1^-$, we can write the shared
information as $\mathcal{I}_1 = I[\CausalState_0^+ ; \CausalState_1^-]$.
The bidirectional machine provides access to the joint distribution
$\Pr(\ForwardCausalState_0, \ReverseCausalState_0, \MeasSymbol_0,
     \ForwardCausalState_1, \ReverseCausalState_1)$ and from this, we can
calculate $\mathcal{I}_1$ in closed-form.

Finally, Refs.~\cite{Abdallah2010} and \cite{Abdallah2010a} investigated the
binding information
$b_\mu = I[\MeasSymbol_0 ; \MeasSymbol_{:1} | \MeasSymbol_{:0}]$ and
the residual entropy
$r_\mu = H[\MeasSymbol_0 | \MeasSymbol_{:0}, \MeasSymbol_{1:}]$. There, they
had to be computed essentially by brute force. Fortunately, the bidirectional
machine again allows us to compute these exactly and in a manner similar to
that for $\mathcal{I}_1$. We again replace $\MeasSymbol_{:0}$ by
$\CausalState_0^+$ and $\MeasSymbol_{1:}$ by $\CausalState_1^-$, giving $b_\mu =
I[\MeasSymbol_0 ; \CausalState_1^- | \CausalState_0^+]$ and $r_\mu =
H[\MeasSymbol_0 | \CausalState_0^+ , \CausalState_1^-]$. Here also, the
bidirectional machine's transitions provide the joint distribution
$\Prob(\CausalState_0^+, \CausalState_0^-, \MeasSymbol_0, \CausalState_1^+,
\CausalState_1^-)$, which can be manipulated appropriately to compute both
$b_\mu$ and $r_\mu$.

In summary, we see that the bidirectional machine gives ready access to
closed-form calculations for a wide range of measures in complex processes.

\subsection{Example}

We close by returning to the irreversible example of
Fig.~\ref{fig:IrreversibleExample}.  Its forward \eM\ has two causal states
while its reverse \eM\ has three causal states. When the partitions for each
\eM\ are logically \textbf{AND}ed together, we obtain the bidirectional machine's
partition over bi-infinite strings.

A compelling visualization of the bidirectional machine's partition is to
superpose the partitions that appeared in Fig.~\ref{fig:IrreversibleExample}.
For example, in the forward \eM, the square corresponding to $X_{-1}X_0 = 21$
was associated with state $A$ (turquoise).  In the reverse \eM, the same
square was associated with state $D$ (red-orange).  Together, the same square
appears in Fig.~\ref{fig:IrreversibleExampleBiMachine}, as both $A$ and $D$.

\begin{figure*}
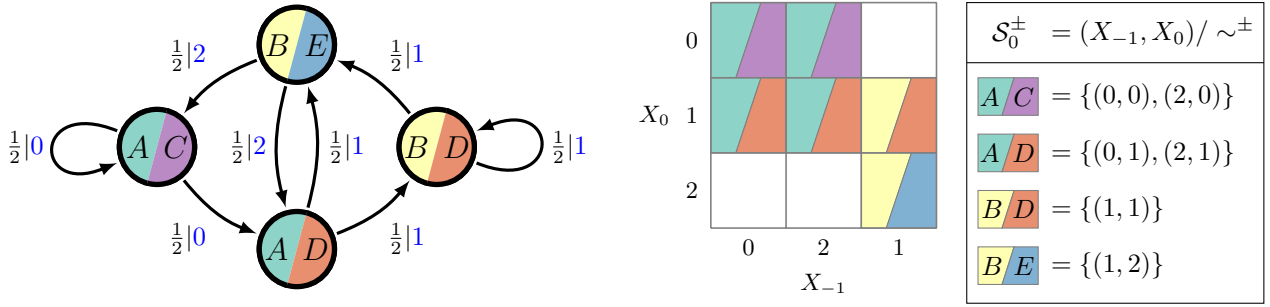

  \begin{minipage}{\columnwidth}
  \centering
  
  \ifgenerate
    \ifcache
      \tikzsetnextfilename{irrev2to3_bieM}
    \fi
    \input{\tikzpathirrev2to3_bieM.tikz}
  \else
    \includegraphics{\tikzpathexternalirrev2to3_bieM}
  \fi

  \end{minipage}%
  \begin{minipage}{\columnwidth}
  
  \ifgenerate
    \ifcache
      \tikzsetnextfilename{irrev2to3_bieMpartition}
    \fi
    \input{\tikzpathirrev2to3_bieMpartition.tikz}
  \else
    \includegraphics{\tikzpathexternalirrev2to3_bieMpartition}
  \fi

  \end{minipage}
\caption{The bidirectional machine $M^\pm$ (left) has causal states
  $\BiCausalState = (\ForwardCausalState, \ReverseCausalState)$ that partition
  bi-infinite sequences $\MeasSymbol_{:}$ of the causally irreversible process
  of Fig.~\ref{fig:IrreversibleExample} (right). In this case, it is sufficient
  to partition sequences using only $(X_{-1}, X_0)$. However, when used as a
  forward (or reverse) generator, the states of the resulting hidden Markov
  model $M^\pm$ do not correspond to a partition of the pasts (or futures)
  since the machine is nonunifilar, as is directly checked in the
  state-transition diagram.
  }
\label{fig:IrreversibleExampleBiMachine}
\end{figure*}

Continuing superposition, we see that there are four bidirectional states and
that these four states partition all bi-infinite sequences.
In particular, bidirectional state $AD$ includes any sequences ending with a $0$
or $2$ and beginning with a $1$. So, if one learns $\Cmu^\pm$ bits, then one
has the luxury, in this case, of knowing that the next symbol \emph{must} be
a $1$. There is inherent uncertainty in the retrodicting the previous symbol.
This is easily verified in the bidirectional machine $M^\pm$ (left) of
Fig.~\ref{fig:IrreversibleExampleBiMachine}.

Finally, Fig.~\ref{fig:BidirectionalIDiagram} presents the bidirectional
machine's information diagram sans the past $H[X_{:0}]$ and future $H[X_{0:}]$.
The three circles, now drawn to scale, represent the statistical complexities
for the forward and reverse \eMs\ and, also, for the bidirectional machine.
Note that the bidirectional machine's state is simply the combination of the
forward and reverse causal states. Calculations give $\ForwardCmu = 1$ bit,
$\ReverseCmu = 3/2$ bit, $\EE = 1/2$ bit, $\Cmu^\pm = 2$ bits.  This yields
$\PC^+ = \PC(M^\pm) = \OI(M^\mp) = 1/2$ bit, $\PC^- = \PC(M^\mp) =
\OI(M^\pm) = 1$ bit and, finally, $\GI = 0$,
verifying Eq.~(\ref{eq:BiMachineNoGauge}).  The bidirectional machine
also gives $\mathcal{I}_1 = 0$, $b_\mu = 1/2$, and $r_\mu = 1/2$.  Then,
according to Ref. \cite{Jame11a}, the entropy rate is
$\hmu = b_\mu + r_\mu = 1$.

\begin{figure}
\centering
\includegraphics{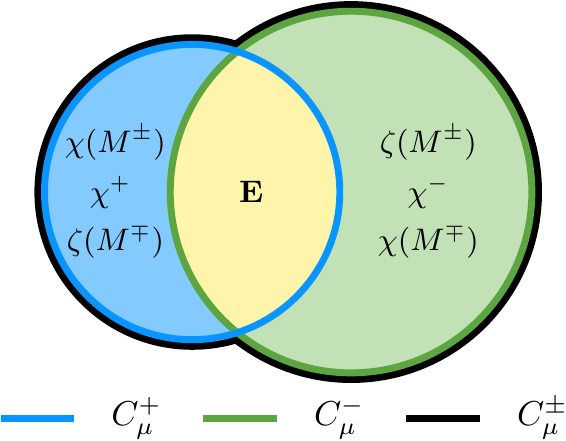}
\caption{
  A quantitatively scaled information diagram for the bidirectional machine of
  Fig.~\ref{fig:IrreversibleExampleBiMachine}. The bidirectional states combine
  the forward and reverse causal states and are represented by the black,
  encompassing line.  Since the forward (blue) and reverse (green) statistical
  complexities lay completely within the past and future respectively, the
  bidirectional machine has no gauge information: $\GI(M^\pm) = 0$.}
\label{fig:BidirectionalIDiagram}
\end{figure}

\section{Conclusion}

The preceding developed a rather thorough survey of reversibility,
irreversibility, and time asymmetry---these being understood in the sense of
analyzing a process's statistical and structural properties scanned either in
forward or in reverse directions with respect to the direction in which it was
given or generated. One result was a stark distinction between Markov chains
and hidden Markov models. For one, we explored the ability of hidden Markov
models to finitely represent infinite-state Markov chains. This came at a high
cost, as we noted: The problem of representational degeneracy appears. We
removed this, however, and so were able to present a number of constructive
results by using the \eM\ as a canonical presentation. Considering that our
field of interest is stationary processes, what we found was surprising. First,
irreversibility is a dominant property in process space. Second,
processes that are finite in one direction can explode into infinite-state
processes in the other. And, third, there is a suite of
information-theoretic measures, helpfully and constructively
captured in various information diagrams, that quantitatively distinguish
structural properties of presentations.

The net result is a new appreciation of irreversibility and a new
toolkit for analyzing irreversible processes. There are many
interesting implications of the long list of technical results. To
suggest what these might be and how they will be applied in the
near future, we would like to close by returning to the physical
motivations called out at the beginning. Specifically, we will
comment on the physical meaning of ``hidden'' processes, the
relationship between the diverse irreversibility properties of
processes and possible physical instantiations, and, finally,
irreversibility in thermodynamic processes.

Why hidden processes? During an interaction between any two systems, only a
portion of each system's internal configuration (or state) is presented to
or is available from the other. On the flip side, not every system
can take on the full state information of another. In effect, each system
views the other as a hidden process. Moreover, in this view measurement is only a
special case of interaction. The measurement act typically does not provide
all of the observed system's state. Thus, for measured processes or collections
of interacting systems one should view them and analyze them as inherently
hidden processes.

Although the analysis largely stayed at the level of probability, statistics,
and information, any implementation resides in a physical substrate. This
simple observation leads one to immediately ask, How are the statistical and
structural properties and classifications of irreversibility related to the
organization of a physical substrate? The direct technical answer is that each
atom in the process's information-measure sigma algebra is associated with
particular degrees of freedom, structures, and behaviors in a physical
implementation. The connection can be made constructively: One of the
longest-standing methods to map between continuous-state physical systems
and sequences is given by symbolic dynamics \cite{Lind95a}.

In light of the preceding structural classifications, one now
sees that the range of alternate presentations for a process parallels
and constrains the range of its possible physical
implementations. In this, each different presentation comes with its own
distinct set of properties---redundancy, crypticity, oracular information,
and the like. In short, then, to study a process's presentations, to classify them,
and to metrize their properties is to study fundamental properties of the
associated physical implementations.

Of course, more is required to complete the mapping from a presentation's
intrinsic computation to the required physics. For example, what is the
entailed dissipation? This reminds one, naturally, of Landauer's Principle:
A computation's logical irreversibility is a lower bound on the required
amount of energy dissipation in the physical implementation \cite{Land89a}.
To the extent that dynamical irreversibility and crypticity control logical
irreversibility, then they also put a lower bound on the physical
implementation's rate of energy dissipation. More generally, the development
above gives a qualitative lower bound on the richness available and a wide
range of applications.

As noted in the Introduction, irreversibility is commonly interpreted as a transient
relaxation process. For example, isolated thermodynamic systems move to
equilibrium since, according to Boltzmann, there are overwhelmingly more
microstates associated with the equilibrium macrostate. This is concisely
monitored via the increase in thermodynamic entropy during relaxation from an
ordered state. It is enshrined in the Second Law of Thermodynamics. However,
as we showed, relaxation is not the only kind of irreversibility that a
thermodynamic system can exhibit. There are also irreversibilities, as we
analyzed in detail, \emph{within} nonequilibrium steady states or, equivalently,
within general stationary stochastic processes. There, a
thermodynamic system is still a process, behaving in time. It is the structure
of this temporal behavior that leads to \emph{dynamical irreversibility} within
the set of configuration trajectories---the temporally invariant set consistent
with being in a nonequilibrium steady state. The preceding gave a new view of
just what these structures are, what irreversibility means in hidden processes,
and a general classification scheme for dynamically reversible and
irreversible processes.

Concretely, recent explorations of thermodynamic irreversibility and energy
dissipation \cite{Jarz97a,Croo98a,Toya10a} ignore distinctions that are
critical for properly identifying statistical irreversibility and intrinsic
computation, as laid out here. Thus, the preceding developments provide a
detailed analysis that will help these efforts by rectifying and grounding
these notions, particularly in terms of the possible physical instantiations
of dynamical irreversibility.

Our analysis of how the past and future are contained in the present is
addressed to a complex world in which structure and randomness co-exist:
\begin{quote}
Time present and time past\\
Are both perhaps present in time future,\\
And time future contained in time past. \\
\\
T. S. Eliot, Buirnt Norton, No. 1 of \emphasis{Four Quartets}.
\end{quote}
In considering general stochastic processes, though, the analysis moves substantially beyond the deterministic world of
Laplace's omniscient Daemon, where initial data is exactly preserved for all
times, past and future. Eliot aptly summarizes our exploration of irreversible
processes, their pasts and futures, and the role the bidirectional machine
plays in capturing the structured present.

\section*{Acknowledgments}

We thank Jason Barnett and Nick Travers for helpful discussions and comments.
JR was supported by a Fellowship Computational Sciences of the Volkswagen
Foundation. This work was partially supported by NSF Grant No. PHY-0748828 and
supported by the Defense Advanced Research Projects Agency (DARPA) Physical
Intelligence Subcontract No. 9060-000709. The views, opinions, and findings
contained in this article are those of the authors and should not be
interpreted as representing the official views or policies, either expressed or
implied, of the DARPA or the Department of Defense.

%\bibliography{ref,chaos}

\appendix

%%% AUTOMATED INPUT %%%
%\input{markovresults}
\section{Markov Order is Time Symmetric}
\label{app:markov}

The principal goal here is to review the properties of Markov
processes so that we can establish the time-symmetry of the Markov
order.
  
\begin{Def}
A process $\Process$ is \emph{\ORDER\MOrder{} Markov} if and only if:
\begin{align}
  \Pr(X_0|X_{:0}) = \Pr(X_0|X_{-R:0}) ~.
\end{align}
\end{Def}

If $\Process$ is \ORDER\MOrder{} Markov, then it is also
\ORDER$\MOrder{}^\prime$ Markov for $R^\prime \geq R$.  However, it is common
to refer to the smallest such $R$ as the \emph{Markov order}.

\begin{Lem}
If a process $\Process$ is \ORDER\MOrder{} Markov, then the future depends
only on the last $R$ symbols; that is,
\begin{align}
  \Pr(X_{0:L}|X_{:0}) = \Pr(X_{0:L}|X_{-R:0}) ~.
\end{align}
\end{Lem}

\begin{proof}
By a simple application of the chain rule, we have:
\begin{align*}
   \Pr(X_{0:L} | X_{:0}) &= \prod_{t=0}^L \Pr(X_t|X_{:t}) \\
                         &= \prod_{t=0}^L \Pr(X_t|X_{-R:0},X_{0:t}) \\
                         &= \Pr(X_{0:L} | X_{-R:0}) ~. \qedhere
\end{align*}
\end{proof}

The result generalizes. The probability of any combination of random
variables in the future given the entire past is the same as when given
only the last $R$ symbols.

Note that the Markov definition is not time symmetric. This invites another
notion of Markovity.

\begin{Def}
A process $\Process$ is \emph{\ORDER\MOrder{} reverse-Markov} if and only if:
\begin{align}
   \Pr(X_{-1}|X_{0:}) = \Pr(X_{-1}|X_{0:R}) ~.
\end{align}
\end{Def}

\begin{Lem}
If a process $\Process$ is \ORDER\MOrder{} reverse-Markov, then the past
depends only on the first $R$ symbols:
\begin{align}
   \Pr(X_{-L:0}|X_{0:}) = \Pr(X_{-L:0}|X_{0:R}) ~.
\end{align}
\end{Lem}

It happens that the Markov order and reverse Markov orders are always equal.

\begin{The}
\label{thm:MarkovReversible}
A process $\Process$ is \ORDER\MOrder{} Markov if and only if it is
\ORDER\MOrder{} reverse-Markov.
\end{The}

\begin{proof}
We assume $\Process$ is \ORDER\MOrder{} Markov, and then show that
$\Process$ is \ORDER\MOrder{} reverse-Markov as well. Recall that any joint
distribution can be forward factored as:
\begin{align*}
   \Pr(X_{a:b}) &= \prod_{t=a}^{b-1} \Pr(X_t|X_{a:t}) ~. %&&\text{(forward)}\\
               %&= \prod_{t=a}^{b-1} \Pr(X_t | X_{t+1:b})&&\text{(reverse)} ~.
\end{align*}
If the process is Markovian and $(b-a)>R$, then this factoring simplifies to:
\begin{align*}
   \Pr(X_{a:b}) &= \prod_{\mathclap{t=a}}^{\mathclap{a+R}} \Pr(X_t | X_{a:t})
                   \prod_{\mathclap{u=a+R+1}}^{b-1} \Pr(X_u | X_{u-R:u}) ~.
\end{align*}
Next, we have:
\begin{align*}
   \Pr(X_{-1:L}) &= \prod_{\mathclap{t=-1}}^{\mathclap{R-1}}
                    \Pr(X_t | X_{-1:t})
                    \prod_{\mathclap{u=R}}^{\mathclap{L-1}}
                    \Pr(X_u | X_{u-R:u}) \\
                 &= \Pr(X_{-1:R})
                    \prod_{\mathclap{u=R}}^{\mathclap{L-1}}
                    \Pr(X_u | X_{u-R:u}) \\
                 &= \Pr(X_{-1} | X_{0:R}) \Pr(X_{0:R})
                    \prod_{\mathclap{u=R}}^{\mathclap{L-1}} \Pr(X_u | X_{u-R:u})
\intertext{and}
%\end{align*}
%Similarly,
%\begin{align*}
   \Pr(X_{0:L}) &= \prod_{\mathclap{t=0}}^{\mathclap{R}} \Pr(X_t | X_{0:t})
                   \prod_{\mathclap{u=R+1}}^{\mathclap{L-1}}
                   \Pr(X_u | X_{u-R:u}) \\
                &= \Pr(X_{0:R+1})
                   \prod_{\mathclap{u=R+1}}^{\mathclap{L-1}}
                   \Pr(X_u | X_{u-R:u}) \\
                &= \Pr(X_{0:R}) \prod_{\mathclap{u=R}}^{\mathclap{L-1}}
                   \Pr(X_u | X_{u-R:u}) ~.
\end{align*}
So, finally, we obtain the desired result:
\begin{align*}
   \Pr(X_{-1}|X_{0:L}) & = \frac{ \Pr( X_{-1:L} ) }{ \Pr( X_{0:L} ) } \\
      & = \Pr(X_{-1} | X_{0:R})  ~.
\end{align*}
In the other direction, we use the reverse factoring of a joint distribution:
\begin{align*}
   \Pr(X_{a:b}) &= \prod_{\mathclap{t=a}}^{\mathclap{b-1}}
                   \Pr(X_t | X_{t+1:b}) ~.
\end{align*}
Then, we assume the process is reverse-Markov to obtain:
\begin{align*}
   \Pr(X_{a:b}) &= \prod_{\mathclap{t=a}}^{\mathclap{b-R-2}}
                   \Pr(X_t | X_{t+1:t+1+R})
                   \prod_{\mathclap{u=b-R-1}}^{\mathclap{b-1}}
                   \Pr(X_t | X_{t+1:b}) ~.
\end{align*}
Similarly, we have:
\begin{align*}
   \Pr(X_{-L:1}) & = \Pr(X_0|X_{-R:0})\Pr(X_{-R:0}) \\
   & \qquad \times \prod_{\mathclap{t=-L}}^{\mathclap{-(R+1)}}
   \Pr(X_t | X_{t+1:t+1+R})
\end{align*}
and
\begin{align*}
   \Pr(X_{-L:0}) &= \Pr(X_{-R:0})
                    \prod_{\mathclap{t=-L}}^{\mathclap{-(R+1)}}
                    \Pr(X_t | X_{t+1:t+1+R})  ~.
\end{align*}
Then,
\begin{align*}
   \Pr(X_{0}|X_{-L:0})
      & = \frac{ \Pr( X_{-L:1} ) }{ \Pr( X_{-L:0} ) } \\
      & = \Pr(X_{0} | X_{-R:0}) ~.
\end{align*}
The results hold for every $L>R$ and in the $L \to \infty$ limit, too.
\end{proof}

The two notions of Markovity relate to forward and reverse generators.

\begin{Lem}
\label{lem:MarkovRevMarkov}
The forward generator $M^+$ is \ORDER\MOrder{} Markov if and only if
the reverse generator is \ORDER\MOrder{} reverse-Markov.
\end{Lem}

\begin{proof}
This follows directly from the definition of the reverse process. Assume
$M^+$ has Markov order \MOrder{}. Let $|u|=L-2R$ and $|w|=|v|=R$. Then,
\begin{align*}
   \Pr(\widetilde{X}_{-1} | \widetilde{X}_{0:L}=wuv )
    & = \Pr(X_1 | X_{-L+1:1} = \widetilde{v}\widetilde{u}\widetilde{w}) \\
    & = \Pr(X_1 | X_{-R+1:1} = \widetilde{w}) \\
    & = \Pr(\widetilde{X}_{-1} | \widetilde{X}_{0:R} = w)
	~. \qedhere
\end{align*}
\end{proof}

With this interpretation, it is a short step to see that the Markov order
is reversible.

\begin{Cor}
The forward generator is \ORDER\MOrder{} Markov if and only
the reverse generator is \ORDER\MOrder{} Markov.
\end{Cor}

\begin{proof}
Apply Thm.~\ref{thm:MarkovReversible} and then
Lem.~\ref{lem:MarkovRevMarkov}.
\end{proof}

\section{The Explosive Example Revisited}
\label{app:infiniteeM}

In Sec.~\ref{sec:FiniteToInfinite}, we examined a causally irreversible
process whose forward \eM\ had two causal states, while its reverse \eM\ had
a countable infinity of causal states. Here, we provide details for
calculating this reverse \eM\ from the forward \eM. We give expressions
for the excess entropy and statistical complexities. A detailed analysis of
the various kinds of causal states---recurrent, transient, and elusive---for
the forward and reverse \eMs\ appears in Fig. \ref{fig:infiniteeMfull} and
gives some insight into the origins of the reverse \eM's infinite number of
causal states.

The forward and reverse \eMs\ are shown in Fig.~\ref{fig:infiniteeMfull}.
The entropy rate, since it is reversible \cite{Crut98d}, is easier to
calculate from $M^+$. This is given directly:
\begin{align}
  \hmu &= H[X_0 | \ForwardCausalState_0] \nonumber\\
       &= \frac{3}{5} \log_2 3 + \frac{2}{5}\\
       &\approx 1.350\ 955\ 500\ 432 ~. \nonumber
\end{align}
The forward statistical complexity is:
\begin{align}
  \ForwardCmu &= \frac{3}{5} \log_2 \frac{5}{3}
               + \frac{2}{5} \log_2 \frac{5}{2}\\
              &\approx 0.970\ 950\ 594\ 455 ~. \nonumber
\end{align}
For $n \geq 0$, the mixed-state operator \cite{Crut08b} acting on $\widetilde{M}^\pm$ gives:
\begin{align*}
  \Pr(\ForwardCausalState_0 | \ReverseCausalState_0 = A_{n-1}^-)
  &= \biggl( \frac{3 \cdot 2^n}{3 \cdot 2^k + 2 \cdot 3^n},
             \frac{2 \cdot 3^n}{3 \cdot 2^n + 2 \cdot 3^n} \biggr)\\
\intertext{and}
  \Pr(\ForwardCausalState_0 | \ReverseCausalState_0 = B^-) &= (1,0) ~.
\end{align*}
As it turns out, these distributions are also the mixed states for the
transient causal states of $M^+$ in the basis of its recurrent states.
That is:
\begin{align*}
  \Pr(\ForwardCausalState_0 | \ForwardCausalState_0 = D_n^+)
  &=
  \Pr(\ForwardCausalState_0 | \ReverseCausalState_0 = A_{n-1}^-) ~.
  \end{align*}
To determine $\pi(\ReverseCausalState)$ we solve the following
simultaneous equations:
\begin{align*}
  \pi(B^-) &= \frac{1}{3} \pi(B^-) + \sum_{n=0}^\infty c_n \pi(A_n^-) \\
  \pi(A_0^-)
   &= \frac{2}{3} \pi(B^-) + \sum_{n=0}^\infty a_n \pi(A_n^-) \\
  \pi(A_n^-) &= b_{n-1} \pi(A_{n-1}^-) \qquad n > 0 ~.
\end{align*}
Beginning with the third, we have:
\begin{align*}
  \pi(A_n^-) &= b_{n-1} \pi(A_{n-1}^-)\\
             &= \left(\,\prod_{n=0}^{n-1} b_n \right) \pi(A_0^-) \\
             &= \left(\frac{1 + \frac{2}{3}^n}{2^{n+1}}\right) \pi(A_0^-) ~,
\end{align*}
for $n > 0$.  Then, solving for $\pi(B^-)$, gives:
\begin{align*}
  \pi(B^-) &= \frac{3}{2} \sum_{n=0}^\infty c_n \pi(A_n^-) ~.
\end{align*}
So,
\begin{align*}
  \pi(B^-) &= \frac{3}{4} \pi(A_0^-) ~.
\end{align*}
The normalization constraint becomes:
\begin{align*}
  1 &= \pi(B^-) + \sum_{n=0}^\infty \pi(A_n^-) \\
    &= \frac{3}{4}\pi(A_0^-) + \frac{7}{4} \pi(A_0^-) ~.
\end{align*}
Thus,
\begin{align*}
  \pi(A_n^-) &= \frac{1+\left(\frac{2}{3}\right)^n}{5\cdot 2^n}\\
  \pi(B^-)   &= \frac{3}{10} ~.
\end{align*}
Collecting these together, we find:
\begin{align*}
\ReverseCmu &= \frac{3}{10} \log_2 \frac{10}{3}
            - \sum_{n=0}^\infty
              \left(\frac{1+\left(\frac{2}{3}\right)^n}{5\cdot 2^n}\right)
              \log_2
              \left(\frac{1+\left(\frac{2}{3}\right)^n}{5\cdot 2^n}\right)\\
            &\approx 1.588\ 621\ 621\ 714 ~.
\end{align*}
Finally,
\begin{align*}
\EE &= \ForwardCmu - H[\ForwardCausalState|\PastCausalState]\\
    &= \ForwardCmu - \sum_{n=0}^\infty
                     \left(\frac{1+\left(\frac{2}{3}\right)^n}{5\cdot 2^n}\right)
                     H\left(\frac{3 \cdot 2^n}{3 \cdot 2^k + 2 \cdot 3^n}\right) \\
    &\approx 0.304\ 159\ 734\ 344 ~,
\end{align*}
where $H(\cdot)$ is the binary entropy function.

\begin{figure*}
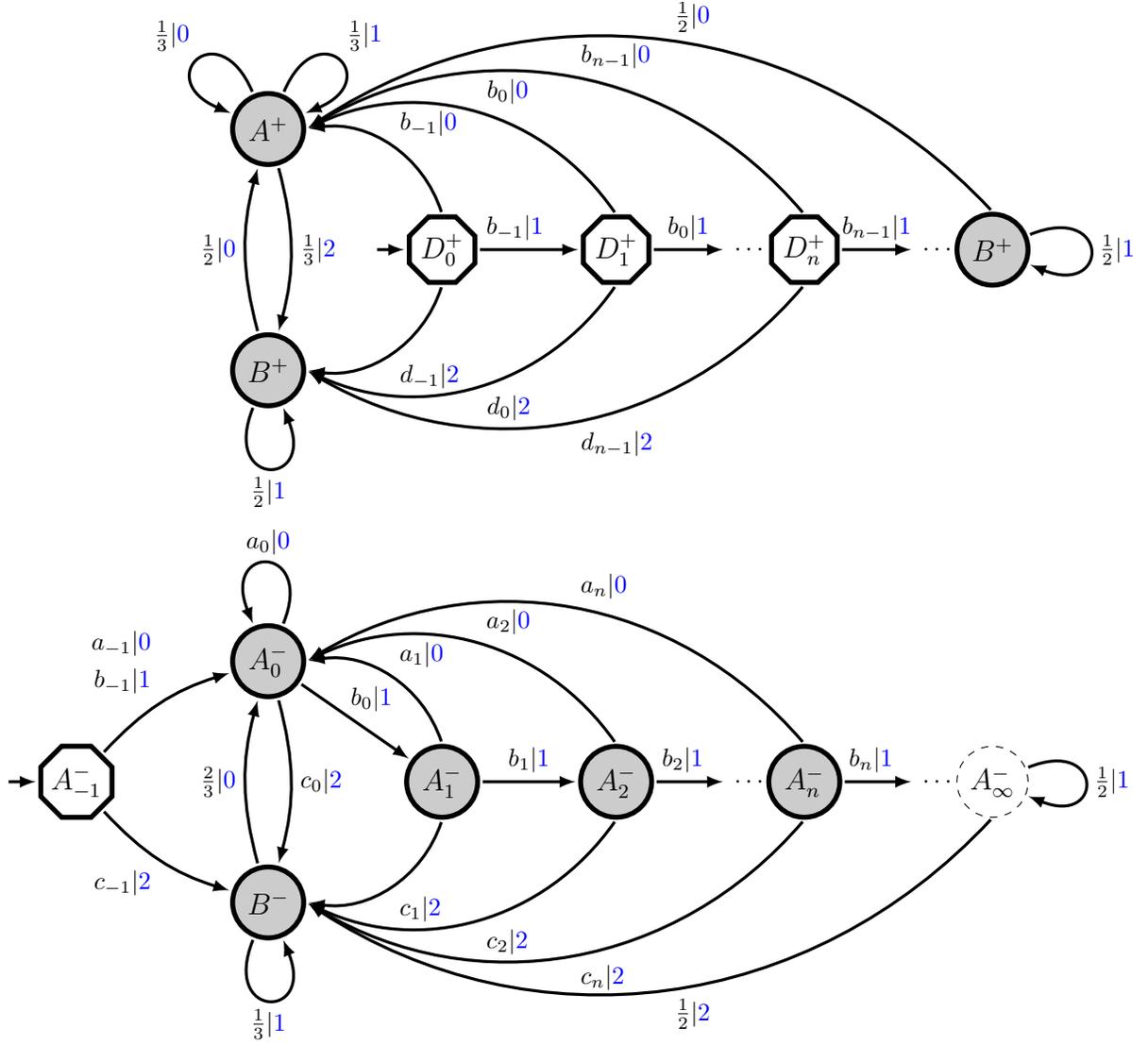


  \ifgenerate
    \ifcache
      \tikzsetnextfilename{explosiveFull}
    \fi
    \input{\tikzpathexplosiveFull.tikz}
  \else
    \includegraphics{\tikzpathexternalexplosiveFull}
  \fi

\caption{The forward \eM\ $M^+$ (top) has only two recurrent (shaded) causal
  states $A^+$ and $B^+$. The reverse \eM\ $M^-$ (bottom) has an infinite
  number of recurrent causal states. Transition labels in both machines make
  use of:
  $a_n = 2^{n+1}(3z_n)^{-1}$,
  $b_n = 1 - (a_n + c_n)$,
  $c_n = 3^n (2z_n)^{-1}$,
  $d_n = 1 - 2b_n$, and
  $z_n = 2^n + 3^n$.
  The dashed state labeled $A_\infty^-$ is an \emph{elusive} causal
  state~\cite{Uppe97a}: It is infinitely preceded, but neither reachable nor
  recurrent.  The hexagon-shaped states are strictly transient states and only
  induced by finite-length histories. Note, the limit of the $D_n^+$ states is
  $D_\infty^+ = B^+$ and it was drawn separately only to demonstrate the trend.
  }
\label{fig:infiniteeMfull}
\end{figure*}

\end{document}